\documentclass{article}
\usepackage{e-jc}
\usepackage{tikz}
\usepackage{caption}
\usepackage{wrapfig}
\usepackage{amsmath}
\usepackage{enumitem}

\usepackage{titling}

\newcommand{\AtLeastOne}{\textsf{AtLeastOne}}
\newcommand{\AtMostOne}{\textsf{AtMostOne}}
\newcommand{\Equal}{\textsf{Equal}}

\definecolor{violet}{RGB}{138,43,226}
\definecolor{forestgreen}{RGB}{34,139,34}
\definecolor{darkblue}{RGB}{102,0,204}
\definecolor{pink}{RGB}{255,192,20}
\definecolor{gold}{RGB}{255,215,0}

\definecolor{cola}{rgb}{0.1, 0.3, 0.6}
\definecolor{colb}{rgb}{0.95, 0.95, 0.95}
\definecolor{colc}{rgb}{0.9, 0.1, 0.4}
\definecolor{cold}{rgb}{0.6, 0.9, 0}
\definecolor{cole}{rgb}{0, 0.7, 0.8}
\definecolor{colf}{rgb}{0.9, 0.6, 0}

\title{Avoiding Monochromatic Rectangles Using Shift Patterns}

\author{Zhenjun Liu  \\
\small Carnegie Mellon University \\
\small {\tt zhenjunl@andrew.cmu.edu}  \and
Leroy Chew\\ 
\small Carnegie Mellon University\\
\small {\tt lchew@andrew.cmu.edu} \and
Marijn J. H. Heule\\
\small Carnegie Mellon University\\
\small {\tt marijn@cmu.edu}
}

\begin{document}

\maketitle

\abstract{

    Ramsey Theory deals with avoiding certain patterns. When constructing an instance that avoids one pattern, it is observed that other patterns emerge. 
    For example, repetition emerges when avoiding arithmetic progression (Van der Waerden numbers), while reflection emerges when avoiding monochromatic solutions of 
    $a+b=c$ (Schur numbers).    
    We exploit observed patterns when coloring a grid while avoiding monochromatic rectangles. Like many problems in Ramsey Theory, this problem has a rapidly growing search
    space that makes computer search difficult. 
    Steinbach et al. obtained a solution of an 18 by 18 grid with 4 colors by enforcing a rotation symmetry. 
    However that symmetry is not suitable for 5 colors. 
    
    In this article, we will encode this problem into propositional logic and enforce so-called internal symmetries, which preserves satisfiability, to guide SAT-solving. 
    We first observe patterns with 2 and 3 colors, among which the ``shift pattern" can be easily generalized and efficiently encoded. 
    Using this pattern, we obtain a new solution of the 18 by 18 grid that is non-isomorphic to the known solution. 
    We further analyze the pattern and obtain necessary conditions to further trim down the search space. 
    We conclude with our attempts on finding a 5-coloring of a 26 by 26 grid, as well as further open problems on the shift pattern.
}

\section{Introduction}
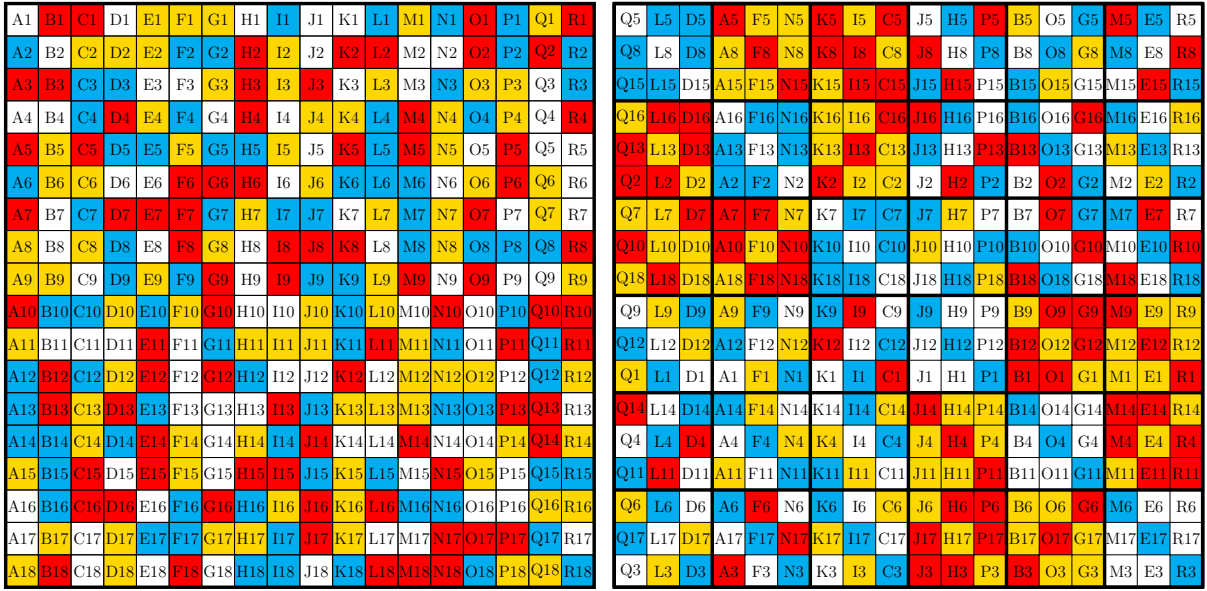
\begin{figure}
    \scalebox{0.43}{
    \begin{tikzpicture}[every node/.style={minimum size=1cm-\pgflinewidth, outer sep=0pt}]
\node[fill=white] at (0.5,17.5){\Large \!A1\!};
\node[fill=cyan] at (0.5,16.5){\Large \!A2\!};
\node[fill=red] at (0.5,15.5){\Large \!A3\!};
\node[fill=white] at (0.5,14.5){\Large \!A4\!};
\node[fill=red] at (0.5,13.5){\Large \!A5\!};
\node[fill=cyan] at (0.5,12.5){\Large \!A6\!};
\node[fill=red] at (0.5,11.5){\Large \!A7\!};
\node[fill=gold] at (0.5,10.5){\Large \!A8\!};
\node[fill=gold] at (0.5,9.5){\Large \!A9\!};
\node[fill=red] at (0.5,8.5){\Large \!A10\!};
\node[fill=gold] at (0.5,7.5){\Large \!A11\!};
\node[fill=cyan] at (0.5,6.5){\Large \!A12\!};
\node[fill=cyan] at (0.5,5.5){\Large \!A13\!};
\node[fill=cyan] at (0.5,4.5){\Large \!A14\!};
\node[fill=gold] at (0.5,3.5){\Large \!A15\!};
\node[fill=white] at (0.5,2.5){\Large \!A16\!};
\node[fill=white] at (0.5,1.5){\Large \!A17\!};
\node[fill=gold] at (0.5,0.5){\Large \!A18\!};
\node[fill=red] at (1.5,17.5){\Large \!B1\!};
\node[fill=white] at (1.5,16.5){\Large \!B2\!};
\node[fill=red] at (1.5,15.5){\Large \!B3\!};
\node[fill=white] at (1.5,14.5){\Large \!B4\!};
\node[fill=gold] at (1.5,13.5){\Large \!B5\!};
\node[fill=gold] at (1.5,12.5){\Large \!B6\!};
\node[fill=white] at (1.5,11.5){\Large \!B7\!};
\node[fill=white] at (1.5,10.5){\Large \!B8\!};
\node[fill=gold] at (1.5,9.5){\Large \!B9\!};
\node[fill=cyan] at (1.5,8.5){\Large \!B10\!};
\node[fill=white] at (1.5,7.5){\Large \!B11\!};
\node[fill=red] at (1.5,6.5){\Large \!B12\!};
\node[fill=red] at (1.5,5.5){\Large \!B13\!};
\node[fill=cyan] at (1.5,4.5){\Large \!B14\!};
\node[fill=cyan] at (1.5,3.5){\Large \!B15\!};
\node[fill=cyan] at (1.5,2.5){\Large \!B16\!};
\node[fill=gold] at (1.5,1.5){\Large \!B17\!};
\node[fill=red] at (1.5,0.5){\Large \!B18\!};
\node[fill=red] at (2.5,17.5){\Large \!C1\!};
\node[fill=gold] at (2.5,16.5){\Large \!C2\!};
\node[fill=cyan] at (2.5,15.5){\Large \!C3\!};
\node[fill=cyan] at (2.5,14.5){\Large \!C4\!};
\node[fill=red] at (2.5,13.5){\Large \!C5\!};
\node[fill=gold] at (2.5,12.5){\Large \!C6\!};
\node[fill=cyan] at (2.5,11.5){\Large \!C7\!};
\node[fill=gold] at (2.5,10.5){\Large \!C8\!};
\node[fill=white] at (2.5,9.5){\Large \!C9\!};
\node[fill=cyan] at (2.5,8.5){\Large \!C10\!};
\node[fill=white] at (2.5,7.5){\Large \!C11\!};
\node[fill=cyan] at (2.5,6.5){\Large \!C12\!};
\node[fill=gold] at (2.5,5.5){\Large \!C13\!};
\node[fill=gold] at (2.5,4.5){\Large \!C14\!};
\node[fill=red] at (2.5,3.5){\Large \!C15\!};
\node[fill=red] at (2.5,2.5){\Large \!C16\!};
\node[fill=white] at (2.5,1.5){\Large \!C17\!};
\node[fill=white] at (2.5,0.5){\Large \!C18\!};
\node[fill=white] at (3.5,17.5){\Large \!D1\!};
\node[fill=gold] at (3.5,16.5){\Large \!D2\!};
\node[fill=cyan] at (3.5,15.5){\Large \!D3\!};
\node[fill=red] at (3.5,14.5){\Large \!D4\!};
\node[fill=cyan] at (3.5,13.5){\Large \!D5\!};
\node[fill=white] at (3.5,12.5){\Large \!D6\!};
\node[fill=red] at (3.5,11.5){\Large \!D7\!};
\node[fill=cyan] at (3.5,10.5){\Large \!D8\!};
\node[fill=cyan] at (3.5,9.5){\Large \!D9\!};
\node[fill=gold] at (3.5,8.5){\Large \!D10\!};
\node[fill=white] at (3.5,7.5){\Large \!D11\!};
\node[fill=gold] at (3.5,6.5){\Large \!D12\!};
\node[fill=red] at (3.5,5.5){\Large \!D13\!};
\node[fill=cyan] at (3.5,4.5){\Large \!D14\!};
\node[fill=white] at (3.5,3.5){\Large \!D15\!};
\node[fill=red] at (3.5,2.5){\Large \!D16\!};
\node[fill=gold] at (3.5,1.5){\Large \!D17\!};
\node[fill=gold] at (3.5,0.5){\Large \!D18\!};
\node[fill=gold] at (4.5,17.5){\Large \!E1\!};
\node[fill=gold] at (4.5,16.5){\Large \!E2\!};
\node[fill=white] at (4.5,15.5){\Large \!E3\!};
\node[fill=gold] at (4.5,14.5){\Large \!E4\!};
\node[fill=cyan] at (4.5,13.5){\Large \!E5\!};
\node[fill=white] at (4.5,12.5){\Large \!E6\!};
\node[fill=red] at (4.5,11.5){\Large \!E7\!};
\node[fill=white] at (4.5,10.5){\Large \!E8\!};
\node[fill=gold] at (4.5,9.5){\Large \!E9\!};
\node[fill=cyan] at (4.5,8.5){\Large \!E10\!};
\node[fill=red] at (4.5,7.5){\Large \!E11\!};
\node[fill=red] at (4.5,6.5){\Large \!E12\!};
\node[fill=cyan] at (4.5,5.5){\Large \!E13\!};
\node[fill=red] at (4.5,4.5){\Large \!E14\!};
\node[fill=red] at (4.5,3.5){\Large \!E15\!};
\node[fill=white] at (4.5,2.5){\Large \!E16\!};
\node[fill=cyan] at (4.5,1.5){\Large \!E17\!};
\node[fill=white] at (4.5,0.5){\Large \!E18\!};
\node[fill=gold] at (5.5,17.5){\Large \!F1\!};
\node[fill=cyan] at (5.5,16.5){\Large \!F2\!};
\node[fill=white] at (5.5,15.5){\Large \!F3\!};
\node[fill=cyan] at (5.5,14.5){\Large \!F4\!};
\node[fill=gold] at (5.5,13.5){\Large \!F5\!};
\node[fill=red] at (5.5,12.5){\Large \!F6\!};
\node[fill=red] at (5.5,11.5){\Large \!F7\!};
\node[fill=red] at (5.5,10.5){\Large \!F8\!};
\node[fill=cyan] at (5.5,9.5){\Large \!F9\!};
\node[fill=gold] at (5.5,8.5){\Large \!F10\!};
\node[fill=white] at (5.5,7.5){\Large \!F11\!};
\node[fill=white] at (5.5,6.5){\Large \!F12\!};
\node[fill=white] at (5.5,5.5){\Large \!F13\!};
\node[fill=gold] at (5.5,4.5){\Large \!F14\!};
\node[fill=gold] at (5.5,3.5){\Large \!F15\!};
\node[fill=cyan] at (5.5,2.5){\Large \!F16\!};
\node[fill=cyan] at (5.5,1.5){\Large \!F17\!};
\node[fill=red] at (5.5,0.5){\Large \!F18\!};
\node[fill=gold] at (6.5,17.5){\Large \!G1\!};
\node[fill=cyan] at (6.5,16.5){\Large \!G2\!};
\node[fill=gold] at (6.5,15.5){\Large \!G3\!};
\node[fill=white] at (6.5,14.5){\Large \!G4\!};
\node[fill=cyan] at (6.5,13.5){\Large \!G5\!};
\node[fill=red] at (6.5,12.5){\Large \!G6\!};
\node[fill=cyan] at (6.5,11.5){\Large \!G7\!};
\node[fill=gold] at (6.5,10.5){\Large \!G8\!};
\node[fill=red] at (6.5,9.5){\Large \!G9\!};
\node[fill=red] at (6.5,8.5){\Large \!G10\!};
\node[fill=cyan] at (6.5,7.5){\Large \!G11\!};
\node[fill=red] at (6.5,6.5){\Large \!G12\!};
\node[fill=white] at (6.5,5.5){\Large \!G13\!};
\node[fill=white] at (6.5,4.5){\Large \!G14\!};
\node[fill=white] at (6.5,3.5){\Large \!G15\!};
\node[fill=red] at (6.5,2.5){\Large \!G16\!};
\node[fill=gold] at (6.5,1.5){\Large \!G17\!};
\node[fill=white] at (6.5,0.5){\Large \!G18\!};
\node[fill=white] at (7.5,17.5){\Large \!H1\!};
\node[fill=red] at (7.5,16.5){\Large \!H2\!};
\node[fill=red] at (7.5,15.5){\Large \!H3\!};
\node[fill=red] at (7.5,14.5){\Large \!H4\!};
\node[fill=cyan] at (7.5,13.5){\Large \!H5\!};
\node[fill=red] at (7.5,12.5){\Large \!H6\!};
\node[fill=gold] at (7.5,11.5){\Large \!H7\!};
\node[fill=white] at (7.5,10.5){\Large \!H8\!};
\node[fill=white] at (7.5,9.5){\Large \!H9\!};
\node[fill=white] at (7.5,8.5){\Large \!H10\!};
\node[fill=gold] at (7.5,7.5){\Large \!H11\!};
\node[fill=cyan] at (7.5,6.5){\Large \!H12\!};
\node[fill=white] at (7.5,5.5){\Large \!H13\!};
\node[fill=gold] at (7.5,4.5){\Large \!H14\!};
\node[fill=red] at (7.5,3.5){\Large \!H15\!};
\node[fill=cyan] at (7.5,2.5){\Large \!H16\!};
\node[fill=gold] at (7.5,1.5){\Large \!H17\!};
\node[fill=cyan] at (7.5,0.5){\Large \!H18\!};
\node[fill=cyan] at (8.5,17.5){\Large \!I1\!};
\node[fill=gold] at (8.5,16.5){\Large \!I2\!};
\node[fill=gold] at (8.5,15.5){\Large \!I3\!};
\node[fill=white] at (8.5,14.5){\Large \!I4\!};
\node[fill=gold] at (8.5,13.5){\Large \!I5\!};
\node[fill=white] at (8.5,12.5){\Large \!I6\!};
\node[fill=cyan] at (8.5,11.5){\Large \!I7\!};
\node[fill=red] at (8.5,10.5){\Large \!I8\!};
\node[fill=red] at (8.5,9.5){\Large \!I9\!};
\node[fill=white] at (8.5,8.5){\Large \!I10\!};
\node[fill=gold] at (8.5,7.5){\Large \!I11\!};
\node[fill=white] at (8.5,6.5){\Large \!I12\!};
\node[fill=red] at (8.5,5.5){\Large \!I13\!};
\node[fill=cyan] at (8.5,4.5){\Large \!I14\!};
\node[fill=red] at (8.5,3.5){\Large \!I15\!};
\node[fill=gold] at (8.5,2.5){\Large \!I16\!};
\node[fill=cyan] at (8.5,1.5){\Large \!I17\!};
\node[fill=cyan] at (8.5,0.5){\Large \!I18\!};
\node[fill=white] at (9.5,17.5){\Large \!J1\!};
\node[fill=white] at (9.5,16.5){\Large \!J2\!};
\node[fill=red] at (9.5,15.5){\Large \!J3\!};
\node[fill=gold] at (9.5,14.5){\Large \!J4\!};
\node[fill=white] at (9.5,13.5){\Large \!J5\!};
\node[fill=gold] at (9.5,12.5){\Large \!J6\!};
\node[fill=cyan] at (9.5,11.5){\Large \!J7\!};
\node[fill=red] at (9.5,10.5){\Large \!J8\!};
\node[fill=cyan] at (9.5,9.5){\Large \!J9\!};
\node[fill=gold] at (9.5,8.5){\Large \!J10\!};
\node[fill=gold] at (9.5,7.5){\Large \!J11\!};
\node[fill=white] at (9.5,6.5){\Large \!J12\!};
\node[fill=cyan] at (9.5,5.5){\Large \!J13\!};
\node[fill=red] at (9.5,4.5){\Large \!J14\!};
\node[fill=cyan] at (9.5,3.5){\Large \!J15\!};
\node[fill=red] at (9.5,2.5){\Large \!J16\!};
\node[fill=red] at (9.5,1.5){\Large \!J17\!};
\node[fill=white] at (9.5,0.5){\Large \!J18\!};
\node[fill=white] at (10.5,17.5){\Large \!K1\!};
\node[fill=red] at (10.5,16.5){\Large \!K2\!};
\node[fill=white] at (10.5,15.5){\Large \!K3\!};
\node[fill=gold] at (10.5,14.5){\Large \!K4\!};
\node[fill=red] at (10.5,13.5){\Large \!K5\!};
\node[fill=cyan] at (10.5,12.5){\Large \!K6\!};
\node[fill=white] at (10.5,11.5){\Large \!K7\!};
\node[fill=red] at (10.5,10.5){\Large \!K8\!};
\node[fill=cyan] at (10.5,9.5){\Large \!K9\!};
\node[fill=cyan] at (10.5,8.5){\Large \!K10\!};
\node[fill=cyan] at (10.5,7.5){\Large \!K11\!};
\node[fill=red] at (10.5,6.5){\Large \!K12\!};
\node[fill=gold] at (10.5,5.5){\Large \!K13\!};
\node[fill=white] at (10.5,4.5){\Large \!K14\!};
\node[fill=gold] at (10.5,3.5){\Large \!K15\!};
\node[fill=gold] at (10.5,2.5){\Large \!K16\!};
\node[fill=gold] at (10.5,1.5){\Large \!K17\!};
\node[fill=cyan] at (10.5,0.5){\Large \!K18\!};
\node[fill=cyan] at (11.5,17.5){\Large \!L1\!};
\node[fill=red] at (11.5,16.5){\Large \!L2\!};
\node[fill=gold] at (11.5,15.5){\Large \!L3\!};
\node[fill=cyan] at (11.5,14.5){\Large \!L4\!};
\node[fill=cyan] at (11.5,13.5){\Large \!L5\!};
\node[fill=cyan] at (11.5,12.5){\Large \!L6\!};
\node[fill=gold] at (11.5,11.5){\Large \!L7\!};
\node[fill=white] at (11.5,10.5){\Large \!L8\!};
\node[fill=gold] at (11.5,9.5){\Large \!L9\!};
\node[fill=gold] at (11.5,8.5){\Large \!L10\!};
\node[fill=red] at (11.5,7.5){\Large \!L11\!};
\node[fill=white] at (11.5,6.5){\Large \!L12\!};
\node[fill=gold] at (11.5,5.5){\Large \!L13\!};
\node[fill=white] at (11.5,4.5){\Large \!L14\!};
\node[fill=cyan] at (11.5,3.5){\Large \!L15\!};
\node[fill=red] at (11.5,2.5){\Large \!L16\!};
\node[fill=white] at (11.5,1.5){\Large \!L17\!};
\node[fill=red] at (11.5,0.5){\Large \!L18\!};
\node[fill=gold] at (12.5,17.5){\Large \!M1\!};
\node[fill=white] at (12.5,16.5){\Large \!M2\!};
\node[fill=white] at (12.5,15.5){\Large \!M3\!};
\node[fill=red] at (12.5,14.5){\Large \!M4\!};
\node[fill=red] at (12.5,13.5){\Large \!M5\!};
\node[fill=cyan] at (12.5,12.5){\Large \!M6\!};
\node[fill=cyan] at (12.5,11.5){\Large \!M7\!};
\node[fill=cyan] at (12.5,10.5){\Large \!M8\!};
\node[fill=red] at (12.5,9.5){\Large \!M9\!};
\node[fill=white] at (12.5,8.5){\Large \!M10\!};
\node[fill=gold] at (12.5,7.5){\Large \!M11\!};
\node[fill=gold] at (12.5,6.5){\Large \!M12\!};
\node[fill=gold] at (12.5,5.5){\Large \!M13\!};
\node[fill=red] at (12.5,4.5){\Large \!M14\!};
\node[fill=white] at (12.5,3.5){\Large \!M15\!};
\node[fill=cyan] at (12.5,2.5){\Large \!M16\!};
\node[fill=white] at (12.5,1.5){\Large \!M17\!};
\node[fill=red] at (12.5,0.5){\Large \!M18\!};
\node[fill=cyan] at (13.5,17.5){\Large \!N1\!};
\node[fill=white] at (13.5,16.5){\Large \!N2\!};
\node[fill=cyan] at (13.5,15.5){\Large \!N3\!};
\node[fill=gold] at (13.5,14.5){\Large \!N4\!};
\node[fill=gold] at (13.5,13.5){\Large \!N5\!};
\node[fill=white] at (13.5,12.5){\Large \!N6\!};
\node[fill=gold] at (13.5,11.5){\Large \!N7\!};
\node[fill=gold] at (13.5,10.5){\Large \!N8\!};
\node[fill=white] at (13.5,9.5){\Large \!N9\!};
\node[fill=red] at (13.5,8.5){\Large \!N10\!};
\node[fill=cyan] at (13.5,7.5){\Large \!N11\!};
\node[fill=gold] at (13.5,6.5){\Large \!N12\!};
\node[fill=cyan] at (13.5,5.5){\Large \!N13\!};
\node[fill=white] at (13.5,4.5){\Large \!N14\!};
\node[fill=red] at (13.5,3.5){\Large \!N15\!};
\node[fill=cyan] at (13.5,2.5){\Large \!N16\!};
\node[fill=red] at (13.5,1.5){\Large \!N17\!};
\node[fill=red] at (13.5,0.5){\Large \!N18\!};
\node[fill=red] at (14.5,17.5){\Large \!O1\!};
\node[fill=red] at (14.5,16.5){\Large \!O2\!};
\node[fill=gold] at (14.5,15.5){\Large \!O3\!};
\node[fill=cyan] at (14.5,14.5){\Large \!O4\!};
\node[fill=white] at (14.5,13.5){\Large \!O5\!};
\node[fill=gold] at (14.5,12.5){\Large \!O6\!};
\node[fill=red] at (14.5,11.5){\Large \!O7\!};
\node[fill=cyan] at (14.5,10.5){\Large \!O8\!};
\node[fill=red] at (14.5,9.5){\Large \!O9\!};
\node[fill=white] at (14.5,8.5){\Large \!O10\!};
\node[fill=white] at (14.5,7.5){\Large \!O11\!};
\node[fill=gold] at (14.5,6.5){\Large \!O12\!};
\node[fill=cyan] at (14.5,5.5){\Large \!O13\!};
\node[fill=white] at (14.5,4.5){\Large \!O14\!};
\node[fill=gold] at (14.5,3.5){\Large \!O15\!};
\node[fill=white] at (14.5,2.5){\Large \!O16\!};
\node[fill=red] at (14.5,1.5){\Large \!O17\!};
\node[fill=cyan] at (14.5,0.5){\Large \!O18\!};
\node[fill=cyan] at (15.5,17.5){\Large \!P1\!};
\node[fill=cyan] at (15.5,16.5){\Large \!P2\!};
\node[fill=gold] at (15.5,15.5){\Large \!P3\!};
\node[fill=gold] at (15.5,14.5){\Large \!P4\!};
\node[fill=red] at (15.5,13.5){\Large \!P5\!};
\node[fill=red] at (15.5,12.5){\Large \!P6\!};
\node[fill=white] at (15.5,11.5){\Large \!P7\!};
\node[fill=cyan] at (15.5,10.5){\Large \!P8\!};
\node[fill=white] at (15.5,9.5){\Large \!P9\!};
\node[fill=cyan] at (15.5,8.5){\Large \!P10\!};
\node[fill=red] at (15.5,7.5){\Large \!P11\!};
\node[fill=white] at (15.5,6.5){\Large \!P12\!};
\node[fill=red] at (15.5,5.5){\Large \!P13\!};
\node[fill=gold] at (15.5,4.5){\Large \!P14\!};
\node[fill=white] at (15.5,3.5){\Large \!P15\!};
\node[fill=white] at (15.5,2.5){\Large \!P16\!};
\node[fill=red] at (15.5,1.5){\Large \!P17\!};
\node[fill=gold] at (15.5,0.5){\Large \!P18\!};
\node[fill=gold] at (16.5,17.5){\Large \!Q1\!};
\node[fill=red] at (16.5,16.5){\Large \!Q2\!};
\node[fill=white] at (16.5,15.5){\Large \!Q3\!};
\node[fill=white] at (16.5,14.5){\Large \!Q4\!};
\node[fill=white] at (16.5,13.5){\Large \!Q5\!};
\node[fill=gold] at (16.5,12.5){\Large \!Q6\!};
\node[fill=gold] at (16.5,11.5){\Large \!Q7\!};
\node[fill=cyan] at (16.5,10.5){\Large \!Q8\!};
\node[fill=white] at (16.5,9.5){\Large \!Q9\!};
\node[fill=red] at (16.5,8.5){\Large \!Q10\!};
\node[fill=cyan] at (16.5,7.5){\Large \!Q11\!};
\node[fill=cyan] at (16.5,6.5){\Large \!Q12\!};
\node[fill=red] at (16.5,5.5){\Large \!Q13\!};
\node[fill=red] at (16.5,4.5){\Large \!Q14\!};
\node[fill=cyan] at (16.5,3.5){\Large \!Q15\!};
\node[fill=gold] at (16.5,2.5){\Large \!Q16\!};
\node[fill=cyan] at (16.5,1.5){\Large \!Q17\!};
\node[fill=gold] at (16.5,0.5){\Large \!Q18\!};
\node[fill=red] at (17.5,17.5){\Large \!R1\!};
\node[fill=cyan] at (17.5,16.5){\Large \!R2\!};
\node[fill=cyan] at (17.5,15.5){\Large \!R3\!};
\node[fill=red] at (17.5,14.5){\Large \!R4\!};
\node[fill=white] at (17.5,13.5){\Large \!R5\!};
\node[fill=white] at (17.5,12.5){\Large \!R6\!};
\node[fill=white] at (17.5,11.5){\Large \!R7\!};
\node[fill=red] at (17.5,10.5){\Large \!R8\!};
\node[fill=gold] at (17.5,9.5){\Large \!R9\!};
\node[fill=red] at (17.5,8.5){\Large \!R10\!};
\node[fill=red] at (17.5,7.5){\Large \!R11\!};
\node[fill=gold] at (17.5,6.5){\Large \!R12\!};
\node[fill=white] at (17.5,5.5){\Large \!R13\!};
\node[fill=gold] at (17.5,4.5){\Large \!R14\!};
\node[fill=cyan] at (17.5,3.5){\Large \!R15\!};
\node[fill=gold] at (17.5,2.5){\Large \!R16\!};
\node[fill=white] at (17.5,1.5){\Large \!R17\!};
\node[fill=cyan] at (17.5,0.5){\Large \!R18\!};
\draw[thick, step=1cm, color=black] (0,0) grid (18, 18);
\draw[thick, step=1cm, color=black] (0,0) grid (18, 18);
\draw[line width=3pt, color=black] (0,0) -- (0,18) -- (18,18) -- (18,0) --cycle;
\end{tikzpicture}
}
\scalebox{0.43}{
\begin{tikzpicture}[every node/.style={minimum size=1cm-\pgflinewidth, outer sep=0pt}]
\node[fill=white] at (0.5,17.5){\Large \!Q5\!};
\node[fill=cyan] at (0.5,16.5){\Large \!Q8\!};
\node[fill=cyan] at (0.5,15.5){\Large \!Q15\!};
\node[fill=gold] at (0.5,14.5){\Large \!Q16\!};
\node[fill=red] at (0.5,13.5){\Large \!Q13\!};
\node[fill=red] at (0.5,12.5){\Large \!Q2\!};
\node[fill=gold] at (0.5,11.5){\Large \!Q7\!};
\node[fill=red] at (0.5,10.5){\Large \!Q10\!};
\node[fill=gold] at (0.5,9.5){\Large \!Q18\!};
\node[fill=white] at (0.5,8.5){\Large \!Q9\!};
\node[fill=cyan] at (0.5,7.5){\Large \!Q12\!};
\node[fill=gold] at (0.5,6.5){\Large \!Q1\!};
\node[fill=red] at (0.5,5.5){\Large \!Q14\!};
\node[fill=white] at (0.5,4.5){\Large \!Q4\!};
\node[fill=cyan] at (0.5,3.5){\Large \!Q11\!};
\node[fill=gold] at (0.5,2.5){\Large \!Q6\!};
\node[fill=cyan] at (0.5,1.5){\Large \!Q17\!};
\node[fill=white] at (0.5,0.5){\Large \!Q3\!};
\node[fill=cyan] at (1.5,17.5){\Large \!L5\!};
\node[fill=white] at (1.5,16.5){\Large \!L8\!};
\node[fill=cyan] at (1.5,15.5){\Large \!L15\!};
\node[fill=red] at (1.5,14.5){\Large \!L16\!};
\node[fill=gold] at (1.5,13.5){\Large \!L13\!};
\node[fill=red] at (1.5,12.5){\Large \!L2\!};
\node[fill=gold] at (1.5,11.5){\Large \!L7\!};
\node[fill=gold] at (1.5,10.5){\Large \!L10\!};
\node[fill=red] at (1.5,9.5){\Large \!L18\!};
\node[fill=gold] at (1.5,8.5){\Large \!L9\!};
\node[fill=white] at (1.5,7.5){\Large \!L12\!};
\node[fill=cyan] at (1.5,6.5){\Large \!L1\!};
\node[fill=white] at (1.5,5.5){\Large \!L14\!};
\node[fill=cyan] at (1.5,4.5){\Large \!L4\!};
\node[fill=red] at (1.5,3.5){\Large \!L11\!};
\node[fill=cyan] at (1.5,2.5){\Large \!L6\!};
\node[fill=white] at (1.5,1.5){\Large \!L17\!};
\node[fill=gold] at (1.5,0.5){\Large \!L3\!};
\node[fill=cyan] at (2.5,17.5){\Large \!D5\!};
\node[fill=cyan] at (2.5,16.5){\Large \!D8\!};
\node[fill=white] at (2.5,15.5){\Large \!D15\!};
\node[fill=red] at (2.5,14.5){\Large \!D16\!};
\node[fill=red] at (2.5,13.5){\Large \!D13\!};
\node[fill=gold] at (2.5,12.5){\Large \!D2\!};
\node[fill=red] at (2.5,11.5){\Large \!D7\!};
\node[fill=gold] at (2.5,10.5){\Large \!D10\!};
\node[fill=gold] at (2.5,9.5){\Large \!D18\!};
\node[fill=cyan] at (2.5,8.5){\Large \!D9\!};
\node[fill=gold] at (2.5,7.5){\Large \!D12\!};
\node[fill=white] at (2.5,6.5){\Large \!D1\!};
\node[fill=cyan] at (2.5,5.5){\Large \!D14\!};
\node[fill=red] at (2.5,4.5){\Large \!D4\!};
\node[fill=white] at (2.5,3.5){\Large \!D11\!};
\node[fill=white] at (2.5,2.5){\Large \!D6\!};
\node[fill=gold] at (2.5,1.5){\Large \!D17\!};
\node[fill=cyan] at (2.5,0.5){\Large \!D3\!};
\node[fill=red] at (3.5,17.5){\Large \!A5\!};
\node[fill=gold] at (3.5,16.5){\Large \!A8\!};
\node[fill=gold] at (3.5,15.5){\Large \!A15\!};
\node[fill=white] at (3.5,14.5){\Large \!A16\!};
\node[fill=cyan] at (3.5,13.5){\Large \!A13\!};
\node[fill=cyan] at (3.5,12.5){\Large \!A2\!};
\node[fill=red] at (3.5,11.5){\Large \!A7\!};
\node[fill=red] at (3.5,10.5){\Large \!A10\!};
\node[fill=gold] at (3.5,9.5){\Large \!A18\!};
\node[fill=gold] at (3.5,8.5){\Large \!A9\!};
\node[fill=cyan] at (3.5,7.5){\Large \!A12\!};
\node[fill=white] at (3.5,6.5){\Large \!A1\!};
\node[fill=cyan] at (3.5,5.5){\Large \!A14\!};
\node[fill=white] at (3.5,4.5){\Large \!A4\!};
\node[fill=gold] at (3.5,3.5){\Large \!A11\!};
\node[fill=cyan] at (3.5,2.5){\Large \!A6\!};
\node[fill=white] at (3.5,1.5){\Large \!A17\!};
\node[fill=red] at (3.5,0.5){\Large \!A3\!};
\node[fill=gold] at (4.5,17.5){\Large \!F5\!};
\node[fill=red] at (4.5,16.5){\Large \!F8\!};
\node[fill=gold] at (4.5,15.5){\Large \!F15\!};
\node[fill=cyan] at (4.5,14.5){\Large \!F16\!};
\node[fill=white] at (4.5,13.5){\Large \!F13\!};
\node[fill=cyan] at (4.5,12.5){\Large \!F2\!};
\node[fill=red] at (4.5,11.5){\Large \!F7\!};
\node[fill=gold] at (4.5,10.5){\Large \!F10\!};
\node[fill=red] at (4.5,9.5){\Large \!F18\!};
\node[fill=cyan] at (4.5,8.5){\Large \!F9\!};
\node[fill=white] at (4.5,7.5){\Large \!F12\!};
\node[fill=gold] at (4.5,6.5){\Large \!F1\!};
\node[fill=gold] at (4.5,5.5){\Large \!F14\!};
\node[fill=cyan] at (4.5,4.5){\Large \!F4\!};
\node[fill=white] at (4.5,3.5){\Large \!F11\!};
\node[fill=red] at (4.5,2.5){\Large \!F6\!};
\node[fill=cyan] at (4.5,1.5){\Large \!F17\!};
\node[fill=white] at (4.5,0.5){\Large \!F3\!};
\node[fill=gold] at (5.5,17.5){\Large \!N5\!};
\node[fill=gold] at (5.5,16.5){\Large \!N8\!};
\node[fill=red] at (5.5,15.5){\Large \!N15\!};
\node[fill=cyan] at (5.5,14.5){\Large \!N16\!};
\node[fill=cyan] at (5.5,13.5){\Large \!N13\!};
\node[fill=white] at (5.5,12.5){\Large \!N2\!};
\node[fill=gold] at (5.5,11.5){\Large \!N7\!};
\node[fill=red] at (5.5,10.5){\Large \!N10\!};
\node[fill=red] at (5.5,9.5){\Large \!N18\!};
\node[fill=white] at (5.5,8.5){\Large \!N9\!};
\node[fill=gold] at (5.5,7.5){\Large \!N12\!};
\node[fill=cyan] at (5.5,6.5){\Large \!N1\!};
\node[fill=white] at (5.5,5.5){\Large \!N14\!};
\node[fill=gold] at (5.5,4.5){\Large \!N4\!};
\node[fill=cyan] at (5.5,3.5){\Large \!N11\!};
\node[fill=white] at (5.5,2.5){\Large \!N6\!};
\node[fill=red] at (5.5,1.5){\Large \!N17\!};
\node[fill=cyan] at (5.5,0.5){\Large \!N3\!};
\node[fill=red] at (6.5,17.5){\Large \!K5\!};
\node[fill=red] at (6.5,16.5){\Large \!K8\!};
\node[fill=gold] at (6.5,15.5){\Large \!K15\!};
\node[fill=gold] at (6.5,14.5){\Large \!K16\!};
\node[fill=gold] at (6.5,13.5){\Large \!K13\!};
\node[fill=red] at (6.5,12.5){\Large \!K2\!};
\node[fill=white] at (6.5,11.5){\Large \!K7\!};
\node[fill=cyan] at (6.5,10.5){\Large \!K10\!};
\node[fill=cyan] at (6.5,9.5){\Large \!K18\!};
\node[fill=cyan] at (6.5,8.5){\Large \!K9\!};
\node[fill=red] at (6.5,7.5){\Large \!K12\!};
\node[fill=white] at (6.5,6.5){\Large \!K1\!};
\node[fill=white] at (6.5,5.5){\Large \!K14\!};
\node[fill=gold] at (6.5,4.5){\Large \!K4\!};
\node[fill=cyan] at (6.5,3.5){\Large \!K11\!};
\node[fill=cyan] at (6.5,2.5){\Large \!K6\!};
\node[fill=gold] at (6.5,1.5){\Large \!K17\!};
\node[fill=white] at (6.5,0.5){\Large \!K3\!};
\node[fill=gold] at (7.5,17.5){\Large \!I5\!};
\node[fill=red] at (7.5,16.5){\Large \!I8\!};
\node[fill=red] at (7.5,15.5){\Large \!I15\!};
\node[fill=gold] at (7.5,14.5){\Large \!I16\!};
\node[fill=red] at (7.5,13.5){\Large \!I13\!};
\node[fill=gold] at (7.5,12.5){\Large \!I2\!};
\node[fill=cyan] at (7.5,11.5){\Large \!I7\!};
\node[fill=white] at (7.5,10.5){\Large \!I10\!};
\node[fill=cyan] at (7.5,9.5){\Large \!I18\!};
\node[fill=red] at (7.5,8.5){\Large \!I9\!};
\node[fill=white] at (7.5,7.5){\Large \!I12\!};
\node[fill=cyan] at (7.5,6.5){\Large \!I1\!};
\node[fill=cyan] at (7.5,5.5){\Large \!I14\!};
\node[fill=white] at (7.5,4.5){\Large \!I4\!};
\node[fill=gold] at (7.5,3.5){\Large \!I11\!};
\node[fill=white] at (7.5,2.5){\Large \!I6\!};
\node[fill=cyan] at (7.5,1.5){\Large \!I17\!};
\node[fill=gold] at (7.5,0.5){\Large \!I3\!};
\node[fill=red] at (8.5,17.5){\Large \!C5\!};
\node[fill=gold] at (8.5,16.5){\Large \!C8\!};
\node[fill=red] at (8.5,15.5){\Large \!C15\!};
\node[fill=red] at (8.5,14.5){\Large \!C16\!};
\node[fill=gold] at (8.5,13.5){\Large \!C13\!};
\node[fill=gold] at (8.5,12.5){\Large \!C2\!};
\node[fill=cyan] at (8.5,11.5){\Large \!C7\!};
\node[fill=cyan] at (8.5,10.5){\Large \!C10\!};
\node[fill=white] at (8.5,9.5){\Large \!C18\!};
\node[fill=white] at (8.5,8.5){\Large \!C9\!};
\node[fill=cyan] at (8.5,7.5){\Large \!C12\!};
\node[fill=red] at (8.5,6.5){\Large \!C1\!};
\node[fill=gold] at (8.5,5.5){\Large \!C14\!};
\node[fill=cyan] at (8.5,4.5){\Large \!C4\!};
\node[fill=white] at (8.5,3.5){\Large \!C11\!};
\node[fill=gold] at (8.5,2.5){\Large \!C6\!};
\node[fill=white] at (8.5,1.5){\Large \!C17\!};
\node[fill=cyan] at (8.5,0.5){\Large \!C3\!};
\node[fill=white] at (9.5,17.5){\Large \!J5\!};
\node[fill=red] at (9.5,16.5){\Large \!J8\!};
\node[fill=cyan] at (9.5,15.5){\Large \!J15\!};
\node[fill=red] at (9.5,14.5){\Large \!J16\!};
\node[fill=cyan] at (9.5,13.5){\Large \!J13\!};
\node[fill=white] at (9.5,12.5){\Large \!J2\!};
\node[fill=cyan] at (9.5,11.5){\Large \!J7\!};
\node[fill=gold] at (9.5,10.5){\Large \!J10\!};
\node[fill=white] at (9.5,9.5){\Large \!J18\!};
\node[fill=cyan] at (9.5,8.5){\Large \!J9\!};
\node[fill=white] at (9.5,7.5){\Large \!J12\!};
\node[fill=white] at (9.5,6.5){\Large \!J1\!};
\node[fill=red] at (9.5,5.5){\Large \!J14\!};
\node[fill=gold] at (9.5,4.5){\Large \!J4\!};
\node[fill=gold] at (9.5,3.5){\Large \!J11\!};
\node[fill=gold] at (9.5,2.5){\Large \!J6\!};
\node[fill=red] at (9.5,1.5){\Large \!J17\!};
\node[fill=red] at (9.5,0.5){\Large \!J3\!};
\node[fill=cyan] at (10.5,17.5){\Large \!H5\!};
\node[fill=white] at (10.5,16.5){\Large \!H8\!};
\node[fill=red] at (10.5,15.5){\Large \!H15\!};
\node[fill=cyan] at (10.5,14.5){\Large \!H16\!};
\node[fill=white] at (10.5,13.5){\Large \!H13\!};
\node[fill=red] at (10.5,12.5){\Large \!H2\!};
\node[fill=gold] at (10.5,11.5){\Large \!H7\!};
\node[fill=white] at (10.5,10.5){\Large \!H10\!};
\node[fill=cyan] at (10.5,9.5){\Large \!H18\!};
\node[fill=white] at (10.5,8.5){\Large \!H9\!};
\node[fill=cyan] at (10.5,7.5){\Large \!H12\!};
\node[fill=white] at (10.5,6.5){\Large \!H1\!};
\node[fill=gold] at (10.5,5.5){\Large \!H14\!};
\node[fill=red] at (10.5,4.5){\Large \!H4\!};
\node[fill=gold] at (10.5,3.5){\Large \!H11\!};
\node[fill=red] at (10.5,2.5){\Large \!H6\!};
\node[fill=gold] at (10.5,1.5){\Large \!H17\!};
\node[fill=red] at (10.5,0.5){\Large \!H3\!};
\node[fill=red] at (11.5,17.5){\Large \!P5\!};
\node[fill=cyan] at (11.5,16.5){\Large \!P8\!};
\node[fill=white] at (11.5,15.5){\Large \!P15\!};
\node[fill=white] at (11.5,14.5){\Large \!P16\!};
\node[fill=red] at (11.5,13.5){\Large \!P13\!};
\node[fill=cyan] at (11.5,12.5){\Large \!P2\!};
\node[fill=white] at (11.5,11.5){\Large \!P7\!};
\node[fill=cyan] at (11.5,10.5){\Large \!P10\!};
\node[fill=gold] at (11.5,9.5){\Large \!P18\!};
\node[fill=white] at (11.5,8.5){\Large \!P9\!};
\node[fill=white] at (11.5,7.5){\Large \!P12\!};
\node[fill=cyan] at (11.5,6.5){\Large \!P1\!};
\node[fill=gold] at (11.5,5.5){\Large \!P14\!};
\node[fill=gold] at (11.5,4.5){\Large \!P4\!};
\node[fill=red] at (11.5,3.5){\Large \!P11\!};
\node[fill=red] at (11.5,2.5){\Large \!P6\!};
\node[fill=red] at (11.5,1.5){\Large \!P17\!};
\node[fill=gold] at (11.5,0.5){\Large \!P3\!};
\node[fill=gold] at (12.5,17.5){\Large \!B5\!};
\node[fill=white] at (12.5,16.5){\Large \!B8\!};
\node[fill=cyan] at (12.5,15.5){\Large \!B15\!};
\node[fill=cyan] at (12.5,14.5){\Large \!B16\!};
\node[fill=red] at (12.5,13.5){\Large \!B13\!};
\node[fill=white] at (12.5,12.5){\Large \!B2\!};
\node[fill=white] at (12.5,11.5){\Large \!B7\!};
\node[fill=cyan] at (12.5,10.5){\Large \!B10\!};
\node[fill=red] at (12.5,9.5){\Large \!B18\!};
\node[fill=gold] at (12.5,8.5){\Large \!B9\!};
\node[fill=red] at (12.5,7.5){\Large \!B12\!};
\node[fill=red] at (12.5,6.5){\Large \!B1\!};
\node[fill=cyan] at (12.5,5.5){\Large \!B14\!};
\node[fill=white] at (12.5,4.5){\Large \!B4\!};
\node[fill=white] at (12.5,3.5){\Large \!B11\!};
\node[fill=gold] at (12.5,2.5){\Large \!B6\!};
\node[fill=gold] at (12.5,1.5){\Large \!B17\!};
\node[fill=red] at (12.5,0.5){\Large \!B3\!};
\node[fill=white] at (13.5,17.5){\Large \!O5\!};
\node[fill=cyan] at (13.5,16.5){\Large \!O8\!};
\node[fill=gold] at (13.5,15.5){\Large \!O15\!};
\node[fill=white] at (13.5,14.5){\Large \!O16\!};
\node[fill=cyan] at (13.5,13.5){\Large \!O13\!};
\node[fill=red] at (13.5,12.5){\Large \!O2\!};
\node[fill=red] at (13.5,11.5){\Large \!O7\!};
\node[fill=white] at (13.5,10.5){\Large \!O10\!};
\node[fill=cyan] at (13.5,9.5){\Large \!O18\!};
\node[fill=red] at (13.5,8.5){\Large \!O9\!};
\node[fill=gold] at (13.5,7.5){\Large \!O12\!};
\node[fill=red] at (13.5,6.5){\Large \!O1\!};
\node[fill=white] at (13.5,5.5){\Large \!O14\!};
\node[fill=cyan] at (13.5,4.5){\Large \!O4\!};
\node[fill=white] at (13.5,3.5){\Large \!O11\!};
\node[fill=gold] at (13.5,2.5){\Large \!O6\!};
\node[fill=red] at (13.5,1.5){\Large \!O17\!};
\node[fill=gold] at (13.5,0.5){\Large \!O3\!};
\node[fill=cyan] at (14.5,17.5){\Large \!G5\!};
\node[fill=gold] at (14.5,16.5){\Large \!G8\!};
\node[fill=white] at (14.5,15.5){\Large \!G15\!};
\node[fill=red] at (14.5,14.5){\Large \!G16\!};
\node[fill=white] at (14.5,13.5){\Large \!G13\!};
\node[fill=cyan] at (14.5,12.5){\Large \!G2\!};
\node[fill=cyan] at (14.5,11.5){\Large \!G7\!};
\node[fill=red] at (14.5,10.5){\Large \!G10\!};
\node[fill=white] at (14.5,9.5){\Large \!G18\!};
\node[fill=red] at (14.5,8.5){\Large \!G9\!};
\node[fill=red] at (14.5,7.5){\Large \!G12\!};
\node[fill=gold] at (14.5,6.5){\Large \!G1\!};
\node[fill=white] at (14.5,5.5){\Large \!G14\!};
\node[fill=white] at (14.5,4.5){\Large \!G4\!};
\node[fill=cyan] at (14.5,3.5){\Large \!G11\!};
\node[fill=red] at (14.5,2.5){\Large \!G6\!};
\node[fill=gold] at (14.5,1.5){\Large \!G17\!};
\node[fill=gold] at (14.5,0.5){\Large \!G3\!};
\node[fill=red] at (15.5,17.5){\Large \!M5\!};
\node[fill=cyan] at (15.5,16.5){\Large \!M8\!};
\node[fill=white] at (15.5,15.5){\Large \!M15\!};
\node[fill=cyan] at (15.5,14.5){\Large \!M16\!};
\node[fill=gold] at (15.5,13.5){\Large \!M13\!};
\node[fill=white] at (15.5,12.5){\Large \!M2\!};
\node[fill=cyan] at (15.5,11.5){\Large \!M7\!};
\node[fill=white] at (15.5,10.5){\Large \!M10\!};
\node[fill=red] at (15.5,9.5){\Large \!M18\!};
\node[fill=red] at (15.5,8.5){\Large \!M9\!};
\node[fill=gold] at (15.5,7.5){\Large \!M12\!};
\node[fill=gold] at (15.5,6.5){\Large \!M1\!};
\node[fill=red] at (15.5,5.5){\Large \!M14\!};
\node[fill=red] at (15.5,4.5){\Large \!M4\!};
\node[fill=gold] at (15.5,3.5){\Large \!M11\!};
\node[fill=cyan] at (15.5,2.5){\Large \!M6\!};
\node[fill=white] at (15.5,1.5){\Large \!M17\!};
\node[fill=white] at (15.5,0.5){\Large \!M3\!};
\node[fill=cyan] at (16.5,17.5){\Large \!E5\!};
\node[fill=white] at (16.5,16.5){\Large \!E8\!};
\node[fill=red] at (16.5,15.5){\Large \!E15\!};
\node[fill=white] at (16.5,14.5){\Large \!E16\!};
\node[fill=cyan] at (16.5,13.5){\Large \!E13\!};
\node[fill=gold] at (16.5,12.5){\Large \!E2\!};
\node[fill=red] at (16.5,11.5){\Large \!E7\!};
\node[fill=cyan] at (16.5,10.5){\Large \!E10\!};
\node[fill=white] at (16.5,9.5){\Large \!E18\!};
\node[fill=gold] at (16.5,8.5){\Large \!E9\!};
\node[fill=red] at (16.5,7.5){\Large \!E12\!};
\node[fill=gold] at (16.5,6.5){\Large \!E1\!};
\node[fill=red] at (16.5,5.5){\Large \!E14\!};
\node[fill=gold] at (16.5,4.5){\Large \!E4\!};
\node[fill=red] at (16.5,3.5){\Large \!E11\!};
\node[fill=white] at (16.5,2.5){\Large \!E6\!};
\node[fill=cyan] at (16.5,1.5){\Large \!E17\!};
\node[fill=white] at (16.5,0.5){\Large \!E3\!};
\node[fill=white] at (17.5,17.5){\Large \!R5\!};
\node[fill=red] at (17.5,16.5){\Large \!R8\!};
\node[fill=cyan] at (17.5,15.5){\Large \!R15\!};
\node[fill=gold] at (17.5,14.5){\Large \!R16\!};
\node[fill=white] at (17.5,13.5){\Large \!R13\!};
\node[fill=cyan] at (17.5,12.5){\Large \!R2\!};
\node[fill=white] at (17.5,11.5){\Large \!R7\!};
\node[fill=red] at (17.5,10.5){\Large \!R10\!};
\node[fill=cyan] at (17.5,9.5){\Large \!R18\!};
\node[fill=gold] at (17.5,8.5){\Large \!R9\!};
\node[fill=gold] at (17.5,7.5){\Large \!R12\!};
\node[fill=red] at (17.5,6.5){\Large \!R1\!};
\node[fill=gold] at (17.5,5.5){\Large \!R14\!};
\node[fill=red] at (17.5,4.5){\Large \!R4\!};
\node[fill=red] at (17.5,3.5){\Large \!R11\!};
\node[fill=white] at (17.5,2.5){\Large \!R6\!};
\node[fill=white] at (17.5,1.5){\Large \!R17\!};
\node[fill=cyan] at (17.5,0.5){\Large \!R3\!};
\draw[thick, step=1cm, color=black] (0,0) grid (18, 18);
\draw[line width=3pt, step=3cm, color=black] (0,0) grid (18, 18);
\draw[line width=3pt, color=black] (0,0) -- (0,18) -- (18,18) -- (18,0) --cycle;
\end{tikzpicture}
}
\caption{Solution of 18 by 18 grid found by Steinbach et al.\ and an isomorphic solution with shift pattern.}
\label{figure: steinbach_sol}
\end{figure}

Consider coloring an $m$ by $n$ grid with $k$ colors, such that there is no monochromatic rectangle. When this is possible, we say that the $m$ by $n$ grid is $k$-colorable. Many results have been derived by pure combinatorial approach: for example, a generalization of Van der Waerden's Theorem can give an upper bound, and Fenner et al.~\cite{fenner2010rectangle} showed that
for each prime power $k$, a $k^2 + k$ by $k$ grid is $k$-colorable but adding a row makes it not $k$-colorable. However, these results are unable to decide many grid sizes: whether an 18 by 18 grid is 4-colorable is an example. This grid had been the last missing piece of the question of 4-colorability, and a challenge prize was raised to close the gap in 2009 \cite{Hayes1717}. This innocuous problem turned out to be an extremely complex computational challenge.
Three years later, Steinbach et al.~\cite{steinbach_posthoff_2012} found a valid 4-coloring of that grid by encoding the problem into propositional logic and applying SAT-solving techniques. 
SAT techniques have also been effective for various other mathematical problems, including Erd\H{o}s discrepancy problem~\cite{konev:2014}, the Pythagorean Triples problem~\cite{Ptn}, and Keller's conjecture~\cite{Keller}. 
The solution by Steinbach et al. is shown in Figure~\ref{figure: steinbach_sol}. 
Notice that the color assignments are highly symmetric: assignments of red is obtained by rotating the assignments of white around the center by 90 degrees, blue by 180 degrees, and so on. 
By now, the $k$-colorability has been decided for $k \in \{2, 3, 4\}$ for all grids.

Therefore, it is natural to ask, what about 5 colors? Applying the aforementioned theorem of Fenner et al.~\cite{fenner2010rectangle}, the 25 by 30 grid is 5-colorable, but for other grids such as 26 by 26 the problem remains open. Like many combinatorial search problems, the rectangle-free grid coloring problem is characterized by enormous search space and rich symmetries. Symmetry breaking is a common technique to trim down the search space while preserving satisfiability. While breaking symmetries between different solutions is definitely helpful, breaking the so-called ``internal symmetries" that is within a specific solution has also been proved to be effective~\cite{heule2010symmetry}. 
Enforcing observed patterns is also known as ``streamlining''~\cite{streamlining} and ``resolution tunnels''~\cite{tunnels} and has
been effective to improve lower bounds of various combinatorial problems including Van der Waerden numbers~\cite{tunnels,triples}, Latin squares~\cite{streamlining}, and graceful graphs~\cite{heule2010symmetry}. 

However, the rotation internal symmetry that Steinbach et al.~\cite{steinbach_posthoff_2012} applied cannot translate to 5 colors. In finding a 4-coloring of the 18 by 18 grid, Steinbach et al.~\cite{steinbach_posthoff_2012} generated a ``cyclic reusable assignment" for one color, and rotate the solution by 90, 180, and 270 degrees to assign to the remaining three. While rotation by 90 degrees is a natural symmetry for 4 colors, this does not apply to the number of colors $k$ that are not multiples of 4.
\par
Thus, to find a 5-coloring of 26 by 26, or rather, to find a valid coloring for any number of colors $k$ in general, an internal symmetry that is applicable to all $k$ is very desirable. With these missions in mind, we found a novel internal symmetry that is unrestricted by the number of colors $k$: As Figure~\ref{figure: steinbach_sol} shows, this pattern is compatible with the known Steinbach's solution of 18 by 18 grid, and it leads to new solutions of 18 by 18 grid that cannot be obtained from the existing one. Search time can be further reduced for this pattern: further analysis poses number of occurrences constraints to colors; a correct color distribution solution to these constraints can reduce the search time for $G_{24, 24}$ and $G_{25, 25}$ to a matter of minutes. Equipped with this new pattern, we also made attempts to solve the 26 by 26 grid; many attempts came down to only 2 or 3 unsatisfied clauses, but none of them succeeded.
\par
The remainder of this article is structured as follows: Section~\ref{section:prelim} discusses preliminaries; in Section~\ref{section:class}, as a starting point, we classify all 2-colorings of 4 by 4 grids and all 3-colorings of 10 by 10 grids; Section~\ref{section:pattern} describes some generalizable patterns we discovered and their analysis; Section~\ref{section:discussions} discusses our attempts on the 26 by 26 grid problems, as well as conjectures collected.

\section{Preliminaries}\label{section:prelim}

\paragraph*{Definitions and terminologies}
Let $G_{m, n}$ denote an $m$ by $n$ grid, and label the cells with ordered pairs in the set $[m] \times [n]$. For a fixed coloring of $G_{m, n}$, a \textit{monochromatic rectangle} is four distinct points $(a, b), (a, d), (c, b), (c, d) \in [m] \times [n]$ that are assigned the same color. A grid $G_{m, n}$ is \textit{$k$-colorable} if, for some coloring $\chi: [m] \times [n] \rightarrow [k]$, there is no monochromatic rectangle, and such $\chi$ is called a $k$-coloring.
\par
Throughout this article, we shall use the word ``pattern" and ``internal symmetry" interchangeably.
We use the term ``shift pattern", when not specified, to refer to a pattern with all subgrids shifted in the same direction.

\paragraph*{Preliminary encoding}
Let $G_{m, n}$ be a grid and $r$ be a positive integer. To encode the problem into propositional formulas, let $a_{ij}^r$ be a variable meaning ``the cell on the $i^{\text{th}}$ row and the $j^{\text{th}}$ column is assigned color $r$''.
\par
First, in order to assign one and only one color to each cell, we need to assert the following:
\begin{equation}
    \bigwedge_{i \in [m]} \bigwedge_{j \in [n]}\left(\AtLeastOne\left( \{a_{ij}^c: c \in [k]\}\right) \wedge \AtMostOne \left( \{a_{ij}^c: c \in [k]\}\right)\right)
\end{equation}
where $\AtLeastOne(S) := \bigvee_{x\in S} x$ and $\AtMostOne(S) := \bigwedge_{\{x, y\} \in {S \choose 2}} (\bar{x} \vee \bar{y})$ for a set $S$.
\par
To encode the rectangle-free constraint, we need to assert
\begin{equation}
    \bigwedge_{m_1=1}^{m-j}\bigwedge_{n_1=1}^{n-i}\bigwedge_{c \in [k]} ( \bar{a}^{c}_{i,j} \vee \bar{a}^{c}_{i+n_1,j} \vee \bar{a}^{c}_{i,j+m_1} \vee \bar{a}^{c}_{i+n_1,j+m_1}))
\end{equation}
And taking conjunction of (1) and (2) gives the propositional encoding of ``there is a rectangle-free $k$-coloring of $G_{m, n}$".

\paragraph*{Experimental setup}
Throughout the article, we used the conflict-driven clause learning (CDCL) SAT solver \texttt{CaDiCaL} and 
the local search SAT solver \texttt{YalSAT} developed by Biere~\cite{Biere2019CADICALAT}.
\texttt{CaDiCaL} was used for smaller grid instances of size up to 18 by 18. For larger grid instances with patterns that are likely to be satisfiable, we used \texttt{PalSAT}, a parallel version of \texttt{YalSAT}. In our experience, local search outperforms CDCL on the satisfiable instances, similar to other hard-combinatorial satisfiable problems~\cite{Ptn,matrix}. 
To solve integer constraints 
for color distribution solutions discussed in Section~\ref{section:necessary}, we used \texttt{Z3} Theorem Prover developed by Microsoft \cite{de2008z3}. 
\par
For smaller grid instances of size up to 18 by 18, we ran \texttt{CaDiCaL} on the general purpose linux hosts at Carnegie Mellon University, as these instances tend to be not resource-intensive and have reasonably short runtimes.
\par
For 5-colorings of larger grid instances, the SAT solvers were ran on the Lonestar 5 cluster of Texas at Texas Advance Computing Center, which has Xeon E5-2690 processors with 24 cores per node. We ran \texttt{PalSAT} on one node with option ``-t 12 \$RANDOM\$RANDOM", which specifies 12 worker threads to be created and a random seed.

\section{Classifying colorings of smaller grid examples}\label{section:class}
We consider classifying relatively simpler grid coloring cases to be a good starting point that enables us to gain insight into how many ``essentially different" solutions there are without putting in too much computational resources. The motivation is that, in this problem there are natural symmetries between solutions, and it suffices to consider the representative of each equivalence class. We will use 2-colorings of $G_{4, 4}$ and 3-colorings of $G_{10, 10}$ as examples, as these are maximal squares that are 2 and 3-colorable.

\par
For square grids $G_{n, n}$ and number of colors $k$, there are three categories of natural symmetries between different colorings (*):
\begin{enumerate}
    \item permutation of colors
    \item permutation of rows or columns
    \item transposition, i.e. a flip along the diagonal
\end{enumerate}
One can check that the each of the operations indeed maps one valid coloring to another valid coloring.
\begin{definition}\label{def: iso}
Define an equivalence relation $E$ on $[k]^{[n] \times [n]}$: two colorings are equivalent (or isomorphic) if one can be obtained from the other by applying a sequence of operations from (*).
\end{definition}
 It is routine to check that $E$ is indeed an equivalence relation. Let $Grid(n, n, k)$ be the set of all valid $k$-colorings of $G_{n, n}$. We are interested in the number of equivalence classes of $Grid(4, 4, 2)$ and $Grid(10, 10, 3)$ under $E$.
\par
To count the equivalence classes, one of the natural first thoughts would be to define an appropriate group, such that each orbit corresponds to an equivalence class, and apply Burnside's Lemma. However, there is a drawback to this approach: the identity element of the group acting on $Grid(10, 10, 3)$ fixes all grid colorings, but the exact value of $|Grid(10, 10, 3)|$ is hard to compute.

\par
To avoid the need to enumerate $Grid(10, 10, 3)$, we need to exploit its symmetries more upfront. One idea is to use \texttt{Shatter}
the symmetry breaking tool \cite{shatter}
to add symmetry-breaking clauses to the original CNF formula. This does not alter the number of equivalence classes since symmetric assignments of literals (w.r.t. clause literal graphs) correspond to isomorphic colorings.
\par
On the other hand, we need to efficiently identify isomorphic graph colorings. Our approach is to convert grid colorings to graph colorings, and use \texttt{Bliss}~\cite{JunttilaKaski:ALENEX2007}
graph isomorphism tool to assign to each graph a canonical labeling.
In particular, we identify each member $Grid(n, n, k)$ with a graph on $n^2$ vertices, where each vertex is identified with a cell in the grid and assigned the same label and same color as the cell. Two distinct vertices $(a, b), (c, d)$ are joined by an edge if and only if $a = c$ or $b = d$. In graph $G$, we shall denote this as $(a, b) \sim_G (c, d)$. There are no self-loops on this graph.
\par
In order to reduce the classification of grids to that of graphs, we need to prove that two $k$-colorings of $G_{n, n}$ are isomorphic if and only if their graph representations are isomorphic.

\begin{proposition}
    Let $\chi_1, \chi_2$ be two $k$-colorings of the grid $G_{n, n}$ where $n>1$.
	Let $H_1$ and $H_2$ be the corresponding graph of $\chi_1$ and $\chi_2$, respectively. If $H_1$ and $H_2$ are isomorphic graphs, then $\chi_1$ and $\chi_2$ are isomorphic grid colorings as in definition \ref{def: iso}.
\end{proposition}

\begin{proof}

Denote the vertex sets of $H_1$ and $H_2$ as $V_1 = \{x_{i, j}: 1 \leq i, j \leq n\}, V_2 = \{y_{i, j}: 1 \leq i, j \leq n\}$ respectively. Let $f: V_1 \rightarrow V_2$ be a graph isomorphism. To show that $\chi_1$ and $\chi_2$ are isomorphic colorings, it suffices to argue that rows are mapped to rows and columns are mapped to columns. Notice that this is only the first possibility, the second possibility is that the transposition is in effect and rows are mapped to columns and columns are mapped to rows.

 To check whether the transposition is in effect we pick $x_{i,j},\;x_{i,q}$  two vertices that are in the same row.  By isomorphism $f(x_{i,j})\sim_{H_2} f(x_{i,q})$. This means either:

\begin{enumerate}
	\item $f(x_{i,j})$, $ f(x_{i,q})$ are in the same row \label{case:id}
	\item $f(x_{i,j})$, $ f(x_{i,q})$ are in the same column \label{case:flip}
\end{enumerate}

Intuitively, Case~\ref{case:id} corresponds to the scenario where only row and column permutations are involved in the grid isomorphism, and Case~\ref{case:flip} involves an additional flip along the diagonal. We shall only prove Case 1, since Case~\ref{case:flip} follows from case 1 by applying the additional flip. We will use $x_{i,j},\;x_{i,q}$, and the indices $i,j,q$ throughout the remainder of the proof as fixed.

In Case~\ref{case:id}, to show that the isomorphism only involves row and column permutations, it suffices to argue that for any $1 \leq a, b, d, e \leq n$,
\begin{enumerate}[label=(\roman*)]
    \item $f(x_{a, d})$ and $f(x_{a, e})$ are in the same row \label{case: row}
    \item $f(x_{a, d})$ and $f(x_{b, d})$ are in the same column \label{case: col}
\end{enumerate}

We will prove Case~\ref{case: row}, and Case~\ref{case: col} follows similarly.

\textbf{Case~\ref{case: row}. }Let $f(x_{i,j})$, $ f(x_{i,q})$ be both in row $\alpha$. For any $1 \leq t \leq n$, $f(x_{i,t})$ must also be in row $\alpha$, because $f(x_{i,j}) \sim_{H_2} f(x_{i,t}) \sim_{H_2} f(x_{i,q})$ and $f(x_{i,t})$ cannot share a column with both $f(x_{i,j})$ and $f(x_{i,q})$, since they are on different columns. 
We have argued that row $i$ gets mapped to a row, but what about any other row. Suppose for some $a \neq i $ and some $d, e$ that $f(x_{a,d}), f(x_{a,e})$ are in column $\beta$ (rather than a row). Then every $f(x_{a, t})$ must be in column $\beta$ because it cannot share a row with both $f(x_{a, d})$ and $f(x_{a, e})$. However this column will meet row $\alpha$ at node $y_{\alpha, \beta}$, which shares an edge between $f(x_{a,t})$ and $f(x_{i,t})$ for every $t$. However this is impossible as $i \neq a$. Therefore, for a fixed $1 \leq a \leq n$, all nodes of the form $f(x_{a, d})$ and $f(x_{a, e})$ are in the same row.

\textbf{Case~\ref{case: col}. }For columns, let $f(x_{i,j})$ be in column $\gamma$. By isomorphism for all $t \neq i$, $f(x_{t,j})\sim_{H_2}f(x_{i,j})$.
If $f(x_{t,j})$ is in the same row as $f(x_{i,j})$, then  $f(x_{t,j})$ shares a row with some $f(x_{i,q})$ for some $q \neq j$, which is not allowed. So, $f(x_{t, j})$ is in column $\gamma$ for all $1 \leq t \leq n$.

Now for columns that are not $j$. Suppose for some $d \neq j$ and some $a, b$ that $f(x_{a, d}), f(x_{b, d})$ are both in row $\delta$. Then for all t, $f(x_{t, d})$ are in that row. However node $y_{\delta, \gamma}$ now is adjacent to $f(x_{t, j})$  and $f(x_{t ,d})$ for every $t$, which is again impossible since $d \neq j$.

Since being on the same row and being on the same column are preserved under $f$, the components $f_1, f_2: [n] \rightarrow [n]$ defined as $f(x_{i,j})= x_{f_1(i), f_2(j)}$ are well-defined functions. Note that $f_1$ and $f_2$ must be surjective because $f$ is surjective. Since $[n]$ is finite, $f_1, f_2$ are permutations on rows and columns respectively.

Therefore, there is a color isomorphism between $\chi_1$ and $\chi_2$, given by $\chi_1(x_{i,j}) = \chi_2(x_{f_1(i), f_2(j)})$.
\end{proof}

\begin{proposition}
	Consider a colored $G_{n, n}$ graph $G$, then any of the following induces a graph isomorphism:
	\begin{itemize}
		\item a row transposition
		\item a column transposition
		\item a flip along the diagonal
	\end{itemize}
\end{proposition}

\begin{proof}
	A row transposition $f$ between rows $i_1$ and $i_2$ preserves the columns, so for any $j$ it holds that $f(x_{i_1,j}) \sim f(x_{i_2,j})$.
	For $i\notin \{i_1, i_2\}$ the row is preserved so $f(x_{i,j}) = x_{i,j} \sim x_{i,j'}= f(x_{i,j'})$.
	The argument for column transposition is symmetric.
	The diagonal flip works because columns and rows work in a symmetric way with regards to adjacency in the graph.
	
\end{proof}

Now with this conversion scheme, we can determine whether two grid colorings are isomorphic by comparing their corresponding canonical representatives output by \texttt{Bliss}~\cite{JunttilaKaski:ALENEX2007}, up to color permutations. This is because we consider color permutations on the original grid to be isomorphisms.
\par
We used $Grid(4, 4, 2)$ as a toy example to check the validity of our approach. In this small example, it can be manually shown that there are only 3 solutions up to isomorphism.
We generated all 840 solutions: new solutions are obtained by adding clauses to forbid existing solutions, and after 840 solutions were collected the CNF formula became unsatisfiable. We classified all the solutions by comparing their corresponding graphs using \texttt{Bliss}~\cite{JunttilaKaski:ALENEX2007} as described above. There are three distinct solutions up to isomorphism, as shown in Figure~\ref{figure:4_4_2_reprs}. 

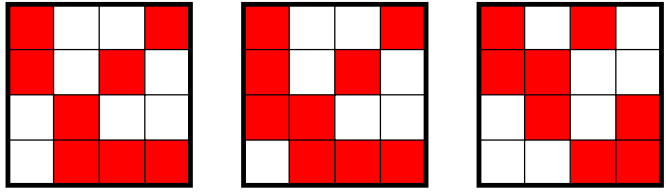
\begin{figure}
    \centering
    \scalebox{0.6}{\begin{tikzpicture}[every node/.style={minimum size=1cm-\pgflinewidth, outer sep=0pt}]
\node[fill=red] at (0.5,3.5) {};
\node[fill=red] at (0.5,2.5) {};
\node[fill=white] at (0.5,1.5) {};
\node[fill=white] at (0.5,0.5) {};
\node[fill=white] at (1.5,3.5) {};
\node[fill=white] at (1.5,2.5) {};
\node[fill=red] at (1.5,1.5) {};
\node[fill=red] at (1.5,0.5) {};
\node[fill=white] at (2.5,3.5) {};
\node[fill=red] at (2.5,2.5) {};
\node[fill=white] at (2.5,1.5) {};
\node[fill=red] at (2.5,0.5) {};
\node[fill=red] at (3.5,3.5) {};
\node[fill=white] at (3.5,2.5) {};
\node[fill=white] at (3.5,1.5) {};
\node[fill=red] at (3.5,0.5) {};
\draw[thick, step=1cm, color=black] (0,0) grid (4, 4);
\draw[line width=3pt, color=black] (0, 0) -- (0, 4) -- (4, 4) -- (4, 0) -- cycle;
\end{tikzpicture}
~~~~~~~
\begin{tikzpicture}[every node/.style={minimum size=1cm-\pgflinewidth, outer sep=0pt}]
\node[fill=red] at (0.5,3.5) {};
\node[fill=red] at (0.5,2.5) {};
\node[fill=red] at (0.5,1.5) {};
\node[fill=white] at (0.5,0.5) {};
\node[fill=white] at (1.5,3.5) {};
\node[fill=white] at (1.5,2.5) {};
\node[fill=red] at (1.5,1.5) {};
\node[fill=red] at (1.5,0.5) {};
\node[fill=white] at (2.5,3.5) {};
\node[fill=red] at (2.5,2.5) {};
\node[fill=white] at (2.5,1.5) {};
\node[fill=red] at (2.5,0.5) {};
\node[fill=red] at (3.5,3.5) {};
\node[fill=white] at (3.5,2.5) {};
\node[fill=white] at (3.5,1.5) {};
\node[fill=red] at (3.5,0.5) {};
\draw[thick, step=1cm, color=black] (0,0) grid (4, 4);
\draw[line width=3pt, color=black] (0, 0) -- (0, 4) -- (4, 4) -- (4, 0) -- cycle;
\end{tikzpicture}
~~~~~~~
\begin{tikzpicture}[every node/.style={minimum size=1cm-\pgflinewidth, outer sep=0pt}]
\node[fill=red] at (0.5,3.5) {};
\node[fill=red] at (0.5,2.5) {};
\node[fill=white] at (0.5,1.5) {};
\node[fill=white] at (0.5,0.5) {};
\node[fill=white] at (1.5,3.5) {};
\node[fill=red] at (1.5,2.5) {};
\node[fill=red] at (1.5,1.5) {};
\node[fill=white] at (1.5,0.5) {};
\node[fill=red] at (2.5,3.5) {};
\node[fill=white] at (2.5,2.5) {};
\node[fill=white] at (2.5,1.5) {};
\node[fill=red] at (2.5,0.5) {};
\node[fill=white] at (3.5,3.5) {};
\node[fill=white] at (3.5,2.5) {};
\node[fill=red] at (3.5,1.5) {};
\node[fill=red] at (3.5,0.5) {};
\draw[thick, step=1cm, color=black] (0,0) grid (4, 4);
\draw[line width=3pt, color=black] (0, 0) -- (0, 4) -- (4, 4) -- (4, 0) -- cycle;
\end{tikzpicture}

}
    \caption{Canonical representative solutions of $Grid(4, 4, 2)$ as output by \texttt{Bliss}}
    \label{figure:4_4_2_reprs}
\end{figure}

In the case of $Grid(10, 10, 3)$, we added symmetry breaking clauses by \texttt{Shatter}, and the CNF formula yielded 35 distinct solutions. Comparison of graph representatives showed that all 35 solutions are isomorphic. That is, there is only one equivalence class in $Grid(10, 10, 3)$; the left grid on Figure \ref{figure:shift_iso} shows the canonical representative output by \texttt{Bliss}~\cite{JunttilaKaski:ALENEX2007}.

\section{Generalizable pattern}
\label{section:pattern}
\subsection{Observing generalizable pattern}

Knowing that all solutions in $Grid(10, 10, 3)$ are isomorphic, it is easier to observe generalizable patterns. This section discusses a pattern involving shifts of rows within subgrids, that is observed in $Grid(10, 10, 3)$ and later larger grids with more colors.
\par
We start by observing shifts in $G_{10, 10}$. Divide $G_{10, 10}$ into 4-subgrids, and within each subgrid, the second row is a copy of the first row except shifted right (or left) by 1; the third row shifted by 2; the last row shifted by 3. All the shifts wrap around on the edge of the subgrid. Figure~\ref{figure:shift_scheme} shows a schematic of this pattern. This is first observed as an isomorphic coloring of the canonical representative of $Grid(10, 10)$; Figure~\ref{figure:shift_iso} shows the representative and the isomorphism side by side. It is noteworthy that this solution also shows a shift pattern on the scale of the entire subgrid, which we shall utilize on larger grids.

\begin{figure}
    \centering
    \scalebox{0.6}{
\begin{tikzpicture}[every node/.style={minimum size=1cm-\pgflinewidth, outer sep=0pt}]
\node[] at (0.5,3.5) {\Large A};
\node[] at (0.5,2.5) {\Large D};
\node[] at (0.5,1.5) {\Large C};
\node[] at (0.5,0.5) {\Large B};
\node[] at (1.5,3.5) {\Large B};
\node[] at (1.5,2.5) {\Large A};
\node[] at (1.5,1.5) {\Large D};
\node[] at (1.5,0.5) {\Large C};
\node[] at (2.5,3.5) {\Large C};
\node[] at (2.5,2.5) {\Large B};
\node[] at (2.5,1.5) {\Large A};
\node[] at (2.5,0.5) {\Large D};
\node[] at (3.5,3.5) {\Large D};
\node[] at (3.5,2.5) {\Large C};
\node[] at (3.5,1.5) {\Large B};
\node[] at (3.5,0.5) {\Large A};
\draw[thick, step=1cm, color=black] (0,0) grid (4, 4);
\draw[line width=3pt, color=black] (0, 0) -- (0, 4) -- (4, 4) -- (4, 0) -- cycle;
\end{tikzpicture}

}
    \caption{Schematic of a 4-subgrid right shift pattern. Each letter denotes a color.}
    \label{figure:shift_scheme}
\end{figure}
\begin{figure}
    \centering
    \scalebox{0.6}{
\begin{tikzpicture}[every node/.style={minimum size=1cm-\pgflinewidth, outer sep=0pt}]
\node[fill=red] at (0.5,9.5) {\Large A1};
\node[fill=cyan] at (0.5,8.5) {\Large A2};
\node[fill=white] at (0.5,7.5) {\Large A3};
\node[fill=white] at (0.5,6.5) {\Large A4};
\node[fill=cyan] at (0.5,5.5) {\Large A5};
\node[fill=red] at (0.5,4.5) {\Large A6};
\node[fill=white] at (0.5,3.5) {\Large A7};
\node[fill=red] at (0.5,2.5) {\Large A8};
\node[fill=cyan] at (0.5,1.5) {\Large A9};
\node[fill=red] at (0.5,0.5) {\Large A0};
\node[fill=white] at (1.5,9.5) {\Large B1};
\node[fill=red] at (1.5,8.5) {\Large B2};
\node[fill=cyan] at (1.5,7.5) {\Large B3};
\node[fill=white] at (1.5,6.5) {\Large B4};
\node[fill=cyan] at (1.5,5.5) {\Large B5};
\node[fill=cyan] at (1.5,4.5) {\Large B6};
\node[fill=red] at (1.5,3.5) {\Large B7};
\node[fill=red] at (1.5,2.5) {\Large B8};
\node[fill=white] at (1.5,1.5) {\Large B9};
\node[fill=cyan] at (1.5,0.5) {\Large B0};
\node[fill=white] at (2.5,9.5) {\Large C1};
\node[fill=red] at (2.5,8.5) {\Large C2};
\node[fill=red] at (2.5,7.5) {\Large C3};
\node[fill=red] at (2.5,6.5) {\Large C4};
\node[fill=white] at (2.5,5.5) {\Large C5};
\node[fill=cyan] at (2.5,4.5) {\Large C6};
\node[fill=cyan] at (2.5,3.5) {\Large C7};
\node[fill=cyan] at (2.5,2.5) {\Large C8};
\node[fill=red] at (2.5,1.5) {\Large C9};
\node[fill=white] at (2.5,0.5) {\Large C0};
\node[fill=red] at (3.5,9.5) {\Large D1};
\node[fill=cyan] at (3.5,8.5) {\Large D2};
\node[fill=red] at (3.5,7.5) {\Large D3};
\node[fill=cyan] at (3.5,6.5) {\Large D4};
\node[fill=red] at (3.5,5.5) {\Large D5};
\node[fill=cyan] at (3.5,4.5) {\Large D6};
\node[fill=white] at (3.5,3.5) {\Large D7};
\node[fill=white] at (3.5,2.5) {\Large D8};
\node[fill=white] at (3.5,1.5) {\Large D9};
\node[fill=white] at (3.5,0.5) {\Large D0};
\node[fill=cyan] at (4.5,9.5) {\Large E1};
\node[fill=cyan] at (4.5,8.5) {\Large E2};
\node[fill=cyan] at (4.5,7.5) {\Large E3};
\node[fill=white] at (4.5,6.5) {\Large E4};
\node[fill=red] at (4.5,5.5) {\Large E5};
\node[fill=white] at (4.5,4.5) {\Large E6};
\node[fill=cyan] at (4.5,3.5) {\Large E7};
\node[fill=red] at (4.5,2.5) {\Large E8};
\node[fill=red] at (4.5,1.5) {\Large E9};
\node[fill=white] at (4.5,0.5) {\Large E0};
\node[fill=white] at (5.5,9.5) {\Large F1};
\node[fill=white] at (5.5,8.5) {\Large F2};
\node[fill=cyan] at (5.5,7.5) {\Large F3};
\node[fill=red] at (5.5,6.5) {\Large F4};
\node[fill=red] at (5.5,5.5) {\Large F5};
\node[fill=white] at (5.5,4.5) {\Large F6};
\node[fill=white] at (5.5,3.5) {\Large F7};
\node[fill=cyan] at (5.5,2.5) {\Large F8};
\node[fill=cyan] at (5.5,1.5) {\Large F9};
\node[fill=red] at (5.5,0.5) {\Large F0};
\node[fill=cyan] at (6.5,9.5) {\Large G1};
\node[fill=white] at (6.5,8.5) {\Large G2};
\node[fill=white] at (6.5,7.5) {\Large G3};
\node[fill=cyan] at (6.5,6.5) {\Large G4};
\node[fill=cyan] at (6.5,5.5) {\Large G5};
\node[fill=red] at (6.5,4.5) {\Large G6};
\node[fill=red] at (6.5,3.5) {\Large G7};
\node[fill=cyan] at (6.5,2.5) {\Large G8};
\node[fill=red] at (6.5,1.5) {\Large G9};
\node[fill=white] at (6.5,0.5) {\Large G0};
\node[fill=cyan] at (7.5,9.5) {\Large H1};
\node[fill=white] at (7.5,8.5) {\Large H2};
\node[fill=red] at (7.5,7.5) {\Large H3};
\node[fill=white] at (7.5,6.5) {\Large H4};
\node[fill=white] at (7.5,5.5) {\Large H5};
\node[fill=cyan] at (7.5,4.5) {\Large H6};
\node[fill=red] at (7.5,3.5) {\Large H7};
\node[fill=white] at (7.5,2.5) {\Large H8};
\node[fill=cyan] at (7.5,1.5) {\Large H9};
\node[fill=red] at (7.5,0.5) {\Large H0};
\node[fill=red] at (8.5,9.5) {\Large I1};
\node[fill=cyan] at (8.5,8.5) {\Large I2};
\node[fill=white] at (8.5,7.5) {\Large I3};
\node[fill=red] at (8.5,6.5) {\Large I4};
\node[fill=white] at (8.5,5.5) {\Large I5};
\node[fill=white] at (8.5,4.5) {\Large I6};
\node[fill=red] at (8.5,3.5) {\Large I7};
\node[fill=cyan] at (8.5,2.5) {\Large I8};
\node[fill=white] at (8.5,1.5) {\Large I9};
\node[fill=cyan] at (8.5,0.5) {\Large I0};
\node[fill=white] at (9.5,9.5) {\Large J1};
\node[fill=red] at (9.5,8.5) {\Large J2};
\node[fill=white] at (9.5,7.5) {\Large J3};
\node[fill=cyan] at (9.5,6.5) {\Large J4};
\node[fill=red] at (9.5,5.5) {\Large J5};
\node[fill=red] at (9.5,4.5) {\Large J6};
\node[fill=cyan] at (9.5,3.5) {\Large J7};
\node[fill=white] at (9.5,2.5) {\Large J8};
\node[fill=cyan] at (9.5,1.5) {\Large J9};
\node[fill=cyan] at (9.5,0.5) {\Large J0};
\draw[thick, step=1cm, color=black] (0,0) grid (10,10);
\draw[line width=3pt, color=black] (0, 0) -- (0, 10) -- (10, 10) -- (10, 0) -- cycle;
\end{tikzpicture}
}
\hspace{0.4cm}
\scalebox{0.6}{
\begin{tikzpicture}[every node/.style={minimum size=1cm-\pgflinewidth, outer sep=0pt}]
\node[fill=red] at (0.5,9.5) {\Large D1};
\node[fill=red] at (0.5,8.5) {\Large D3};
\node[fill=white] at (0.5,7.5) {\Large D0};
\node[fill=cyan] at (0.5,6.5) {\Large D4};
\node[fill=cyan] at (0.5,5.5) {\Large D6};
\node[fill=white] at (0.5,4.5) {\Large D9};
\node[fill=white] at (0.5,3.5) {\Large D8};
\node[fill=cyan] at (0.5,2.5) {\Large D2};
\node[fill=white] at (0.5,1.5) {\Large D7};
\node[fill=red] at (0.5,0.5) {\Large D5};
\node[fill=cyan] at (1.5,9.5) {\Large H1};
\node[fill=red] at (1.5,8.5) {\Large H3};
\node[fill=red] at (1.5,7.5) {\Large H0};
\node[fill=white] at (1.5,6.5) {\Large H4};
\node[fill=cyan] at (1.5,5.5) {\Large H6};
\node[fill=cyan] at (1.5,4.5) {\Large H9};
\node[fill=white] at (1.5,3.5) {\Large H8};
\node[fill=white] at (1.5,2.5) {\Large H2};
\node[fill=red] at (1.5,1.5) {\Large H7};
\node[fill=white] at (1.5,0.5) {\Large H5};
\node[fill=white] at (2.5,9.5) {\Large F1};
\node[fill=cyan] at (2.5,8.5) {\Large F3};
\node[fill=red] at (2.5,7.5) {\Large F0};
\node[fill=red] at (2.5,6.5) {\Large F4};
\node[fill=white] at (2.5,5.5) {\Large F6};
\node[fill=cyan] at (2.5,4.5) {\Large F9};
\node[fill=cyan] at (2.5,3.5) {\Large F8};
\node[fill=white] at (2.5,2.5) {\Large F2};
\node[fill=white] at (2.5,1.5) {\Large F7};
\node[fill=red] at (2.5,0.5) {\Large F5};
\node[fill=red] at (3.5,9.5) {\Large I1};
\node[fill=white] at (3.5,8.5) {\Large I3};
\node[fill=cyan] at (3.5,7.5) {\Large I0};
\node[fill=red] at (3.5,6.5) {\Large I4};
\node[fill=white] at (3.5,5.5) {\Large I6};
\node[fill=white] at (3.5,4.5) {\Large I9};
\node[fill=cyan] at (3.5,3.5) {\Large I8};
\node[fill=cyan] at (3.5,2.5) {\Large I2};
\node[fill=red] at (3.5,1.5) {\Large I7};
\node[fill=white] at (3.5,0.5) {\Large I5};
\node[fill=cyan] at (4.5,9.5) {\Large G1};
\node[fill=white] at (4.5,8.5) {\Large G3};
\node[fill=white] at (4.5,7.5) {\Large G0};
\node[fill=cyan] at (4.5,6.5) {\Large G4};
\node[fill=red] at (4.5,5.5) {\Large G6};
\node[fill=red] at (4.5,4.5) {\Large G9};
\node[fill=cyan] at (4.5,3.5) {\Large G8};
\node[fill=white] at (4.5,2.5) {\Large G2};
\node[fill=red] at (4.5,1.5) {\Large G7};
\node[fill=cyan] at (4.5,0.5) {\Large G5};
\node[fill=cyan] at (5.5,9.5) {\Large E1};
\node[fill=cyan] at (5.5,8.5) {\Large E3};
\node[fill=white] at (5.5,7.5) {\Large E0};
\node[fill=white] at (5.5,6.5) {\Large E4};
\node[fill=white] at (5.5,5.5) {\Large E6};
\node[fill=red] at (5.5,4.5) {\Large E9};
\node[fill=red] at (5.5,3.5) {\Large E8};
\node[fill=cyan] at (5.5,2.5) {\Large E2};
\node[fill=cyan] at (5.5,1.5) {\Large E7};
\node[fill=red] at (5.5,0.5) {\Large E5};
\node[fill=white] at (6.5,9.5) {\Large B1};
\node[fill=cyan] at (6.5,8.5) {\Large B3};
\node[fill=cyan] at (6.5,7.5) {\Large B0};
\node[fill=white] at (6.5,6.5) {\Large B4};
\node[fill=cyan] at (6.5,5.5) {\Large B6};
\node[fill=white] at (6.5,4.5) {\Large B9};
\node[fill=red] at (6.5,3.5) {\Large B8};
\node[fill=red] at (6.5,2.5) {\Large B2};
\node[fill=red] at (6.5,1.5) {\Large B7};
\node[fill=cyan] at (6.5,0.5) {\Large B5};
\node[fill=white] at (7.5,9.5) {\Large J1};
\node[fill=white] at (7.5,8.5) {\Large J3};
\node[fill=cyan] at (7.5,7.5) {\Large J0};
\node[fill=cyan] at (7.5,6.5) {\Large J4};
\node[fill=red] at (7.5,5.5) {\Large J6};
\node[fill=cyan] at (7.5,4.5) {\Large J9};
\node[fill=white] at (7.5,3.5) {\Large J8};
\node[fill=red] at (7.5,2.5) {\Large J2};
\node[fill=cyan] at (7.5,1.5) {\Large J7};
\node[fill=red] at (7.5,0.5) {\Large J5};
\node[fill=white] at (8.5,9.5) {\Large C1};
\node[fill=red] at (8.5,8.5) {\Large C3};
\node[fill=white] at (8.5,7.5) {\Large C0};
\node[fill=red] at (8.5,6.5) {\Large C4};
\node[fill=cyan] at (8.5,5.5) {\Large C6};
\node[fill=red] at (8.5,4.5) {\Large C9};
\node[fill=cyan] at (8.5,3.5) {\Large C8};
\node[fill=red] at (8.5,2.5) {\Large C2};
\node[fill=cyan] at (8.5,1.5) {\Large C7};
\node[fill=white] at (8.5,0.5) {\Large C5};
\node[fill=red] at (9.5,9.5) {\Large A1};
\node[fill=white] at (9.5,8.5) {\Large A3};
\node[fill=red] at (9.5,7.5) {\Large A0};
\node[fill=white] at (9.5,6.5) {\Large A4};
\node[fill=red] at (9.5,5.5) {\Large A6};
\node[fill=cyan] at (9.5,4.5) {\Large A9};
\node[fill=red] at (9.5,3.5) {\Large A8};
\node[fill=cyan] at (9.5,2.5) {\Large A2};
\node[fill=white] at (9.5,1.5) {\Large A7};
\node[fill=cyan] at (9.5,0.5) {\Large A5};
\draw[thick, step=1cm, color=black] (0,0) grid (10,10);
\draw[line width=3pt] (0, 2) -- (8, 2);
\draw[line width=3pt] (0, 6) -- (8, 6);
\draw[line width=3pt] (4, 2) -- (4, 10);
\draw[line width=3pt] (8, 2) -- (8, 10);
\draw[line width=3pt, color=black] (0, 0) -- (0, 10) -- (10, 10) -- (10, 0) -- cycle;
\end{tikzpicture}
}
\caption{Coloring of $G_{10, 10}$ with shift pattern (right) and a visualization of the isomorphism}
\label{figure:shift_iso}
\end{figure}

\begin{table}
\centering
\caption{Experimental results for shift pattern of $Grid(10, 10, 3)$}
\label{table:10_10_3}

\begin{tabular}{ |c|c|c|c|c|c|c|c|c| } 
 \hline
 subgrid size & 2 & 3 & 4 & 5 & 6 & 7 & 8 & 9\\
 \hline
 left shift & unsat & unsat & sat & unsat & unsat & unsat & unsat & unsat\\
 \hline
 both shifts & unsat & unsat & sat & unsat & unsat & unsat & unsat & unsat\\
 \hline
\end{tabular}
\end{table}

Table \ref{table:10_10_3} shows the experimental results for shift pattern for subgrid size $2, 3, ..., 9$ for $Grid(10, 10, 3)$, where in ``left shift" patterns each subsequent row in a subgrid is shifted left, and in ``both shifts" both directions are allowed (so different subgrids could have different shift directions). We only investigated left shifts as a left-shifted subgrid can be converted to a right-shifted subgrid by simply permuting the rows: for example, in a 4-subgrid, the appropriate permutation is swapping the second and last rows. The results demonstrate that whether both directions are allowed generally does not make a difference; we shall see why this is true in the analysis part. On the other hand, it can be mathematically shown that subgrid size 10 is unsatisfiable, as in Proposition \ref{proposition:k2}.

It is also noteworthy that Steinbach's solution of $G_{18, 18}$ is isomorphic to one with shift patterns on size-3 subgrids, as shown in Figure~\ref{figure: steinbach_sol}. The isomorphism is labeled by the letters.

\subsection{Encoding shift pattern}
Encoding the shift pattern constraint is crucial to solving a grid coloring instance efficiently. To encode that ``this subgrid has a left shift pattern", without changing the preliminary encoding in section 2, we add binary clauses for each pair of cells that should be colored the same:
$$\Equal(x, y) := (\bar{x} \vee y) \wedge (x \vee \bar{y})$$
\par
In order to encode that ``this subgrid has either left shift pattern or right shift pattern", a pair of extra variables $(L, R)$ is used for each subgrid. 
Now, let $t_1 = c + \left((q - c + p - r) \text{ mod } b\right)$ and $t_2 = c + \left((q - c - p + r) \text{ mod } b\right)$.

$$\textsf{EqualLeft}(r, c, p, q, l) := \left(  {\bar x_{p,q}^l} \vee x_{r, t_1}^l \vee L \right) \wedge \left( {\bar x_{r, t_1}^l} \vee x_{p,q}^l \vee L \right)$$
$$\textsf{EqualRight}(r, c, p, q, l) := \left( {\bar x_{p,q}^l} \vee x_{r, t_2}^l \vee R \right) \wedge \left( {\bar x_{r, t_2}^l} \vee x_{p,q}^l \vee R \right) $$
$$\textsf{LeftOrRightShift}(r, c):=\!\!\!\!\!\! \bigwedge_{r+1\leq p \leq r+b} \bigwedge_{c\leq q \leq c+b} \bigwedge_{l \in [k]} \textsf{EqualLeft}(r, c, p, q, l) \wedge    \textsf{EqualRight}(r, c, p, q, l) \wedge \textsf{XOR}(L, R)$$
where $(r, c)$ is the index of the top-left cell in the subgrid, $b$ is the size of the subgrid, and $\textsf{XOR}$ is the usual exclusive-or.
\par
However, in either case, the binary and ternary clauses used hinders performance for local search solvers: when local search solvers choose one variable to flip say $x$, the other variable $y$ should be also flipped as well to satisfy $\Equal$. But $y$ is not guaranteed to be flipped as the next variable chosen, creating overhead for satisfying these simple clauses. The same happens for \textsf{EqualLeft} and \textsf{EqualRight}: exactly one of $L, R$ is true, forcing \textsf{EqualLeft} or \textsf{EqualRight} to behave exactly like the $\Equal$ clauses.
\par
The solution to this problem is to modify the preliminary encoding completely. The experiment in Table~\ref{table:10_10_3} shows that there is only need for encoding single-direction shifts, so our solution will target single-direction shifts. Instead of adding binary clauses to enforce that two variables should be assigned the same value, we map the two variables to the same variable. Thus, in each subgrid, only variables for the first row are created, and all the shifted copies on the other rows are referred to as the corresponding variable on the first row. The cells that do not belong to any subgrid are still assigned fresh variables.

\subsection{Necessary conditions for the shift pattern}\label{section:necessary}

We first analyze the necessary condition for a single subgrid to have the shift pattern in Proposition~\ref{proposition:k2} and its corollary.
\begin{proposition}
Let $k$ be a positive integer and $k^2 < n < k^2 + k$. Then there is no $k$-coloring of an $n$-subgrid with the shift pattern.
\label{proposition:k2}
\end{proposition}
\begin{proof}
For any color assignment of the first row, by Pigeonhole Principle, there is a color $c \in [k]$ that occurs at least $k + 1$ times. For a fixed cell $a$ colored $c$ in the first row, there are $k$ subsequent rows in which another cell colored $c$ shares the same column as $a$; let the set of the $k$ rows be denoted $f(a)$. That is, $k$ rows need to be chosen for each of the $k+1$ cells colored $c$ in the first row. Moreover, for two distinct cells $a_1, a_2$ colored $c$ in the first row, $f(a_1) \cap f(a_2) = \varnothing$, since if row $r \in f(a_1) \cap f(a_2)$, then there is a monochromatic rectangle between row 1 and row $r$. However, there are only $k$ disjoint $k$-subsets of the remaining $n-1$ rows. So there is no valid $k$-coloring of this pattern.
\end{proof}
\begin{corollary}
If there is a $k$-coloring of an $n$-subgrid with shift pattern, then $n \leq k^2$ or $n \geq k^2 + k$.
\end{corollary}
\par
On a higher level, necessary conditions on the distribution of colors on the entire grid can be obtained by counting gaps between occurrences of the same color. The length of a gap is defined as the number of cells between the endpoints plus one.
\par
Throughout this section, suppose we have $G_{xz, yz}$ with $k$ colors, divided into square subgrids of length $z$, and the subgrids are indexed by $[x] \times [y]$. When we speak of a ``column" or a ``row", we refer to the $z$ columns or $z$ rows that share subgrids with the same second or first coordinate. Suppose each subgrid has a right shift pattern.
Let $v_{c,i,j}$ denote the number of cells in the subgrid on the $i$th row and the $j$th column that are colored $c$ divided by $z$.
Let $\mathbf{v}_{c,j}$ be the vector $\begin{bmatrix} v_{c,1,j}\\ 
\dots\\ 
v_{c,y,j}
\end{bmatrix}$.

\begin{proposition}
For every column $j$,
$$ \sum_{c\in [k]} \mathbf{v}_{c,j}= \begin{bmatrix} z\\ 
\dots\\ 
z
\end{bmatrix} $$
\label{proposition:sum}
\end{proposition}

\begin{proof}
	This proposition just restates that there are exactly $z^2$ cells in a subgrid.
\end{proof}

The following proposition concerns monochromatic rectangles with two corners in the same subgrid.

\begin{proposition}
	For any color $c$,  any column $j$, 	$$\sum^{y}_{i=1} v_{c,i,j}^2- v_{c,i,j} \leq z-1$$ when $z$ is odd.
	
	$$\sum^{y}_{i=1}v_{c,i,j}^2- v_{c,i,j} \leq z-2$$ when $z$ is even.
	\label{proposition:self}
\end{proposition}

\begin{proof}
	The number of possible gaps is $z-1$, but if any gaps are repeated we get a monochromatic rectangle within the column. Note that due to wrap around of the shift pattern, a gap of length $a$ becomes a gap of length $z - a$. So if $z$ is even we also have to avoid any $\frac{z}{2}$ length gaps, and the limit reduces to $z-2$. 
	
	Consider ``gluing" the left and right edges of every subgrid together, so that the cells of the same color form diagonals. The number of gaps in a single subgrid with $v_{c,i,j}$ many diagonals of color $c$ is $v_{c,i,j}(v_{c,i,j}-1)$, since two gaps are created with each pair of distinct diagonals, and the same gaps cannot be used in any other subgrid on the same column.
\end{proof}

The following proposition concerns monochromatic rectangles with all four corners in four distinct subgrids.

\begin{proposition}
For every color $c$, for every pair of distinct columns $j_1, j_2$,

$$\mathbf{v}_{c,j_1} \cdot \mathbf{v}_{c,j_2} \leq z$$
\label{proposition:scalar_product}
\end{proposition}

\begin{proof}
	Suppose there are $c$-colored monochromatic rectangles between four different blocks.
	These happen if and only if there there are two horizontal gaps of length $l$ along the same columns, which in turn happen if and only if two gaps of equal lengths appear in two pairs of blocks along the same block columns, as the shift pattern will synchronize the gaps to occur in the same places.

	In any pair of columns $j_1$ and $j_2$, this means the number of different possible gaps has an upper limit. Let $\Delta j= |j_1-j_2|$. We detail below (Figure~\ref{figure: pairs}) the possible gaps, most of these come in pairs due to the shift pattern wrapping around.
	
\begin{figure}[h]
	\begin{center}
	\begin{tabular}{|c|c|}
	\hline
	$z\Delta j+ 0$ & \\
	$z\Delta j+ 1$ & $z\Delta j-z+1$\\
	$z\Delta j+ 2$ & $z\Delta j-z+2$\\
	\dots & \dots\\
	$z\Delta j+ z-1$ & $z\Delta j-1$\\
	\hline
	\end{tabular}
\caption{Pairs of possible gaps \label{figure: pairs}}
\end{center}
\end{figure}
	
	This means there is an upper limit of $z$ possible gaps, but the number of gaps along each row $i$ can be up to $v_{c,i,j_1}\cdot v_{c,i,j_2}$.
	
	To show that the same gap cannot appear twice within one row on different pairs of diagonals, we appeal to the fact that a horizontal gap within a subgrid also becomes a vertical gap within a subgrid. This is because there are two sets of diagonals $\{d_1, d_2\}$ and $\{d_3, d_4\}$ with the same horizontal gap, $d_1, d_2$ can both be shifted horizontally to be the places of $d_3$ and $d_4$, but that same shift appears vertically giving a monochromatic rectangle. 
		
	Hence there are $v_{c,i,j_1}\cdot v_{c,i,j_2}$ gaps used in each row.
    The sum of gaps used is the scalar product.
	
\end{proof}


\begin{remark}
We stress that Proposition \ref{proposition:sum}, Proposition \ref{proposition:self} and Proposition \ref{proposition:scalar_product} put together does not give us a sufficient condition. For example, in the case of $G_{16, 16}$ with 4 colors, having $c_{ij} = 1$ for each color and each $i, j$ clearly satisfies the constraints, but experiment shows that there is no such 4-coloring with this color distribution.
\end{remark}
Since the argument for rows are symmetric, Proposition \ref{proposition:sum}, Proposition \ref{proposition:self}, and Proposition \ref{proposition:scalar_product} can be restated and proved in terms of row vectors. Combining these constraint gives solutions for color distributions, and we shall call a set of $k$ matrix solutions for the $k$ colors a ``color distribution solution".

\subsection{New results}
The shift pattern greatly reduces search space, as only the first row in each subgrid needs to be chosen. This pattern can be iterated for one more layer: by adding ``midgrids" that are made up of smaller subgrids and enforcing that subgrids on subsequent rows are shifted copies of those on the first row in a similar fashion, the search space can be further reduced. This is extremely helpful to solving instances of larger grids, and significantly reduces the running time. Figure~\ref{figure:18_18_4_3} shows 4-coloring of $G_{18, 18}$ with 3-subgrids and 9-midgrids. \texttt{CaDiCaL} found this solution in under 1 second; previously, according to Fenner et al.~\cite{fenner2010rectangle}, it took the SAT solver \texttt{clasp} roughly 7 hours to find a cyclic-reusable assignment of $G_{18, 18}$. For 5-coloring of $G_{24, 24}$ with 4-subgrids, 12-midgrids, diagonal and anti-diagonal pattern, \texttt{CaDiCaL} found the solution in Figure~\ref{figure:24_24_5_6} in 2.46 seconds. 
\par
In addition, it is noteworthy that the 4-coloring of $G_{18, 18}$ in Figure~\ref{figure:18_18_4_3} is a ``new" solution, in the sense that it is not isomorphic to the solution given by Fenner et al.~\cite{fenner2010rectangle}. Thus we have reasons to believe that the shift pattern is truly a new pattern.
\begin{figure}
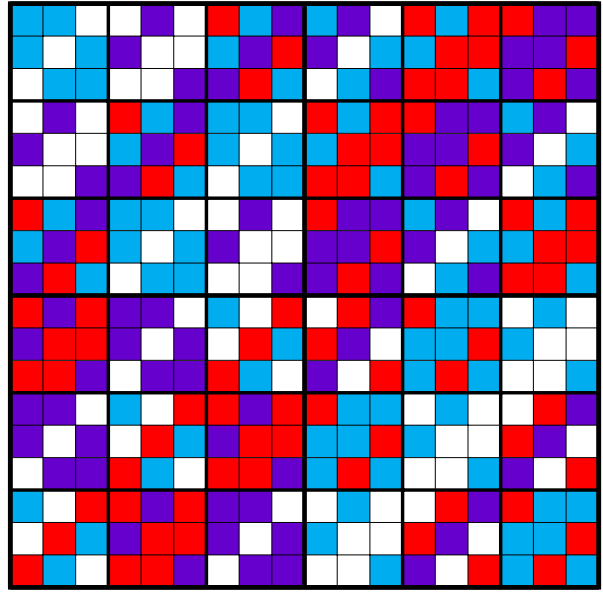

    \centering
\scalebox{0.43}{

}
    \caption{4-coloring of $G_{18, 18}$ with 3-subgrids and 9-midgrids}
    \label{figure:18_18_4_3}
\end{figure}

\begin{figure}
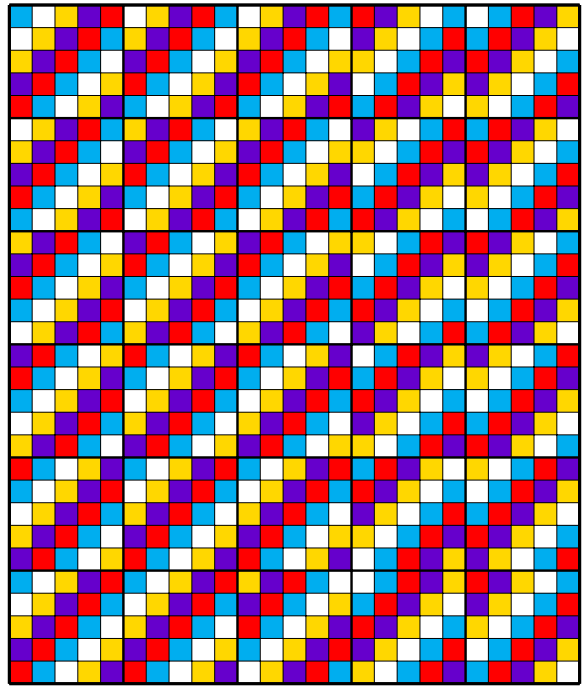

    \centering
\scalebox{0.3}{    
%

}
    \caption{5-coloring of $Grid(30, 25, 5)$ with 5-subgrids}
    \label{figure: 30_25_5_5}
\end{figure}

In addition to shift patterns, a feasible color distribution solution boosts performance for solving large grids significantly. \texttt{PalSAT} could not solve $G_{25, 25}$ with 5-subgrids within 24 hours. With the constraint that $c_{ij} = 1$ for each color and each $i, j$, \texttt{PalSAT} was able to find a solution within 1 second. The solution can be extended to a 30 by 25 grid as shown in Figure~\ref{figure: 30_25_5_5}.

\section{Further discussions on shift pattern}\label{section:discussions}
\subsection{Attempts on $G_{26, 26}$ with shift pattern}
\begin{figure}
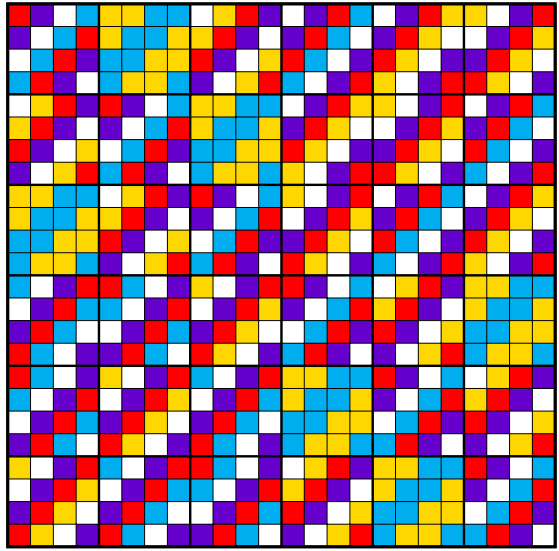

    \centering
    \scalebox{0.3}{

    }
    \caption{5-coloring of $G_{24, 24}$ with 4-subgrids, 12-midgrids, diagonal and anti-diagonal patterns}
    \label{figure:24_24_5_6}
\end{figure}

\begin{figure}
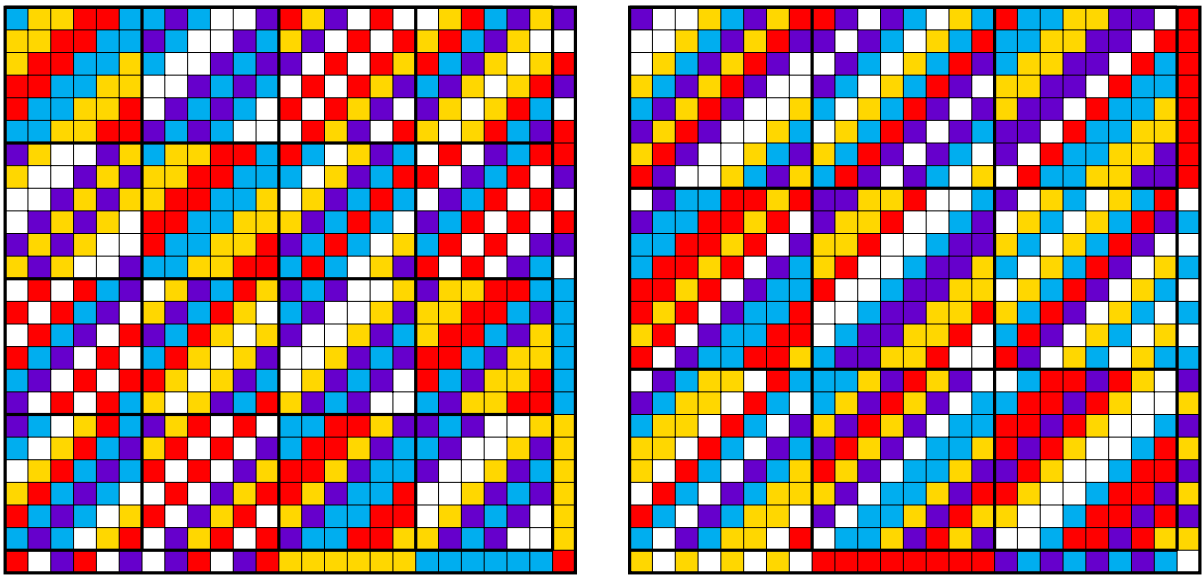

    \centering
    \scalebox{0.3}{
%
}
    \caption{5-colorings of $G_{25, 25}$ with 6-subgrids (left) and 8-subgrids (right)}
    \label{figure: 25_25_5_6and8}
\end{figure}
In this section, we describe our attempts on $G_{26, 26}$ with shift pattern by casing on different subgrid sizes. As we shall see, the case where the entire 26 by 26 grid is divided into subgrids is impossible. Thus, we consider the case where shift pattern is enforced on a smaller grid, and attempt to constrain the remaining rows and columns by partially enforcing the shift pattern.

\subsubsection{Infeasible case: Shift pattern on 26 by 26}
If there are no remaining columns and rows, the only choices left are 2, 13 and 26. By Proposition~\ref{proposition:k2}, we know that the subgrid size cannot be 26. If a 5-coloring of subgrid size 2 or 13 exists, then it must satisfy the necessary conditions in Section~\ref{section:necessary}. For the case of size 2, there exists five 13 by 13 matrices of non-negative integers, where column vector(s) and row vector(s) satisfy Propositions~\ref{proposition:sum}, \ref{proposition:self}, and \ref{proposition:scalar_product}. The similar solutions must exist for analogous constraints for size 13. However, we encoded the constraints into the \texttt{Z3} Theorem Prover and it reported unsatisfiable for both cases. Therefore, no 5-colorings of such shift pattern exist.
\par
Thus, we turn our attention to finding solutions that have shift pattern for the upper-left 25 by 25 or 24 by 24 part, and constrain the remaining column(s) and row(s) as appropriate.
\subsubsection{Shift pattern on 25 by 25}
In the case of 25 by 25, the possible subgrid sizes are 5 and 25.
For subgrid size 5, as mentioned in Section~\ref{section:necessary}, an obvious color distribution is $c_{ij} = 1$ for each color and each $i, j$. We attempted to find a coloring of $G_{26, 26}$ which has this pattern on the 25 by 25 part and left the last row and column without additional constraints. \texttt{PalSAT} was unable to find a satisfying assignment in 24 hours, with only one unsatisfied clause. This instance is unlikely to be satisfiable because of the cell at the bottom-right corner: under this constraint, each 5 by 1 (or 1 by 5) group of cells can be assigned the same color in the last column (or row), but this leaves no legal color for the bottom-right cell, since 5 groups exhaust all colors. There is another possible color distribution solution given by \texttt{Z3}, as shown in Table~\ref{table:25_25_5_5}. However, this constraint is outright unsatisfiable for $G_{26, 26}$, as reported by \texttt{CaDiCaL} in under 10 seconds. We faced a similar difficulty to extend a solution of $G_{25, 25}$ with a 25-subgrid, and this case is harder than 5-subgrids in the sense that we cannot restrict the grid to a non-trivial color distribution constraint, because there is only one subgrid.
\begin{table}[]
    \centering
\begin{tabular}{|@{~\,}c@{~\,}|@{~\,}c@{~\,}|@{~\,}c@{~\,}|@{~\,}c@{~\,}|@{~\,}c@{~\,}|}
\hline
2&1&0&1&1\\
\hline
1&0&2&2&0\\
\hline
1&2&1&1&0\\
\hline
0&1&0&2&1\\
\hline
1&1&2&0&2\\
\hline
\end{tabular}
\,
\begin{tabular}{|@{~\,}c@{~\,}|@{~\,}c@{~\,}|@{~\,}c@{~\,}|@{~\,}c@{~\,}|@{~\,}c@{~\,}|}
\hline
1&0&2&2&1\\
\hline
0&2&1&0&2\\
\hline
2&1&1&0&1\\
\hline
1&2&1&1&0\\
\hline
2&0&0&1&1\\
\hline
\end{tabular}
\,
\begin{tabular}{|@{~\,}c@{~\,}|@{~\,}c@{~\,}|@{~\,}c@{~\,}|@{~\,}c@{~\,}|@{~\,}c@{~\,}|}
\hline
1&1&0&0&2\\
\hline
1&1&0&2&1\\
\hline
0&1&1&0&2\\
\hline
2&0&2&1&1\\
\hline
0&2&2&1&0\\
\hline
\end{tabular}
\,
\begin{tabular}{|@{~\,}c@{~\,}|@{~\,}c@{~\,}|@{~\,}c@{~\,}|@{~\,}c@{~\,}|@{~\,}c@{~\,}|}
\hline
1&2&1&0&1\\
\hline
2&0&2&0&1\\
\hline
0&1&2&2&1\\
\hline
0&1&0&1&2\\
\hline
2&1&0&2&0\\
\hline
\end{tabular}
\,
\begin{tabular}{|@{~\,}c@{~\,}|@{~\,}c@{~\,}|@{~\,}c@{~\,}|@{~\,}c@{~\,}|@{~\,}c@{~\,}|}
\hline
0&1&2&2&0\\
\hline
1&2&0&1&1\\
\hline
2&0&0&2&1\\
\hline
2&1&2&0&1\\
\hline
0&1&1&1&2\\
\hline
\end{tabular}
    \caption{Color distribution solution for $G_{25, 25}$ with 5-subgrid shifts output by \texttt{Z3}}
    \label{table:25_25_5_5}
\end{table}
\subsubsection{Shift pattern on 24 by 24, under further constraints}
This brings us to 24 by 24, which has more possible subgrid sizes. Directly solving $G_{24, 24}$ for subgrid sizes ranging from 3 to 10 shows that subgrid sizes $\{3, 4, 5, 6, 8, 10\}$ are satisfiable.
We also obtain solutions of $G_{25, 25}$ with subgrid sizes 6 and 8, with the additional constraint that all subgrids (or pairs of subgrids) on the diagonal have the same colorings. However, \texttt{PalSAT} was unable to solve $G_{26, 26}$ with the similar shift pattern and diagonal patterns in 24 hours. The reason is likely to be that the unconstrained last two rows and columns add too much complexity to the search space.
\par
Therefore, we sought to constrain them with a ``partial shift pattern": for example, take the upper-left 26 by 26 part of a 32 by 32 grid, for subgrid size 8. In this way, the 25th row under each subgrid is a shifted copy of the 24th row, and the columns are also constrained to be alternating colors, as observed in solutions of $G_{25, 25}$ in Figure~\ref{figure: 25_25_5_6and8}. However, this was not enough: \texttt{PalSAT} could not get under 2 unsatisfied clauses.
\par
So we need further constraints. The most successful constraint so far is to fix the number of occurrences of colors in each subgrid, but can we utilize it in the last two rows?
The idea is to fix the number of occurrences of some color in each ``partial subgrid" in the last two rows. In this way, searching for a 5-coloring of $G_{26, 26}$ is reduced to finding a feasible set of color distribution solutions. In particular, for subgrid sizes in $\{3, 4, 6\}$, the partial subgrids on the last two rows are ``complete", in the sense that, if there is a monochromatic rectangle in a complete subgrid, then there is a monochromatic rectangle in the first two rows; this is because the only possible gaps are of length 1 and 2. Therefore, for sizes 3, 4 and 6, refuting the existence of a set of color distribution solutions for $G_{27, 24}$, $G_{28, 24}$, $G_{30, 24}$ respectively proves that no 5-coloring for $G_{26, 26}$ with such (complete and partial) shift pattern exists. Size 8 does not enjoy this property, because gaps of length 3 are possible. Nonetheless, in this case a color distribution solution is likely to lead to a solution.
\par
Thus, we attempted to prove or disprove the existence of color distribution solutions for sizes in $\{3, 4, 6, 8\}$ via the \texttt{Z3} Theorem Prover. \texttt{Z3} reported unsatisfiable for size 8 in 5 minutes, and did not terminate in 4 days for all other sizes. Directly solving these constraints for an integer solution is perhaps too brute-force, when there is no obvious solution or refutation. 
\subsection{Further open problems}
\subsubsection{Both shifts vs left shifts}
We have seen in Table~\ref{table:10_10_3} that whenever a ``both shifts" solution of a certain subgrid size exists, a ``left shift" solution of the same size exists as well. This leads us to make the following conjecture:
\begin{conjecture}
For $m, n, k, z \in \mathbb{N}^+$, if a $k$-coloring of $G_{m, n}$ with shift pattern in both directions on $z$-subgrids exists, then a  $k$-coloring of $G_{m, n}$ with shift pattern in only one direction on $z$-subgrids exists.
\end{conjecture}
However, the main difficulty in verifying this is the necessary condition for a solution of ``both shifts", which gives the color distribution constraints. It is much more involved than that of ``left shift", because gaps occur differently. Proposition~\ref{proposition:sum} apparently still holds; Proposition~\ref{proposition:self} still holds because allowing different shift directions in the same column mounts to shuffling the appearances of gaps (between appearances of the same color in the same subgrid). The problem is Proposition~\ref{proposition:scalar_product}, where gaps between appearances in different subgrids are considered: when both subgrids are shifted in the same direction, all possible combinations of two colors of the same gap length (modulo subgrid size) are traversed; but when the subgrids are shifted in different directions the traversal does not happen, and to bound the scalar product is much more complex.
\subsection{$k$-coloring of $G_{k^2+k, k}$ with shift pattern}
It is known that $G_{k^2+k, k}$ is $k$-colorable. 
We observed that for $k \in \{2, 3, 5\}$, there is a solution of $G_{k^2, k^2}$ with $k$-subgrids and each color occurring exactly $k$ times in each subgrid (so no color occurs twice in a subgrid-column or subgrid-row). This means that the solution can be extended to a solution of $G_{k^2+k, k}$ by filling the last $k$ rows with ``monochromatic stripes", as illustrated in Figure~\ref{figure: k2}.
\par
Therefore, we make the following conjecture:
\begin{conjecture}
For $k$ prime, there exists a $k$-coloring of $G_{k^2, k^2}$ with shift pattern on $k$-subgrids, such that each color occurs exactly $k$ times in each subgrid.\\
This implies that there is a $k$-coloring of $G_{k^2+k, k}$ with shifted $k$-subgrid on the $k^2$ by $k^2$ part, and monochromatic stripes on the last $k$ rows.
\end{conjecture}
\begin{figure}
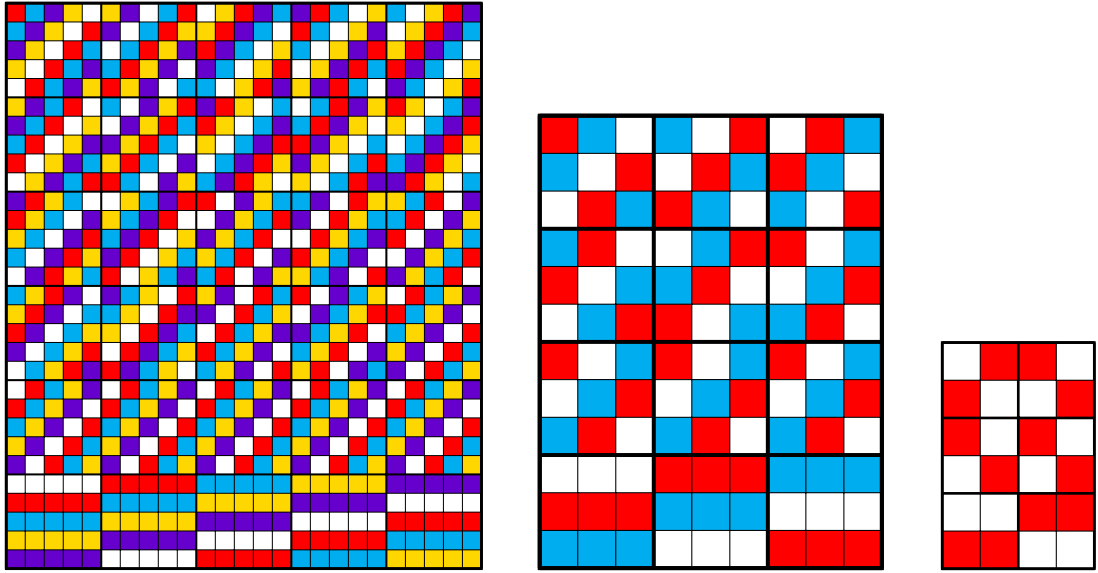

    \centering
\scalebox{0.25}{

}
    \caption{$k$-coloring of $G_{k^2+k, k^2}$ with shift pattern}
    \label{figure: k2}
\end{figure}
\section{Conclusions}
We observed internal symmetries of monochromatic rectangle-free grid colorings, via classifying small grid examples of $G_{10, 10}$ and $G_{4, 4}$ according to solution symmetries. We then generalized the observed patterns; among them, the shift pattern was particularly effective and yielded new solutions for $G_{18, 18}$ with 4 colors. Encoding the shift pattern into SAT can greatly reduce the number of variables and thus the search space. We further analyzed the necessary conditions for a solution with the shift pattern (a single direction) to exist, which poses more effective constraints on the grid by dictating color distribution in each subgrid. Experiment on $G_{30, 25}$ shows that feasible color distribution solution greatly reduces the search time.
\par
Equipped with the shift pattern, we attempted to find a 5-coloring of $G_{26, 26}$. The closest we got was an assignment that falsified two clauses. 
We also made attempts to prove or disprove the existence of a matrix color distribution solution by using the SMT solver \texttt{Z3}.
While the shift pattern seems universal in 3, 4, and 5-colorings, we expect that solving $G_{26, 26}$ may require additional constraints on the rows and columns that do not belong to any subgrids. 
We raise further problems regarding the shift patterns: specifically, regarding the relationship between two and one-directional shift patterns and whether there exists $k-$coloring of $G_{k^2, k^2}$ with shift pattern on $k$-subgrids, such that each color occurs exactly $k$ times in each subgrid. 

\bibliographystyle{plain}
\bibliography{summary_bib}
\end{document}